\documentclass[12pt,a4paper]{article}

\usepackage[utf8]{inputenc}
\usepackage{amsmath,amsthm,amscd,amssymb,eucal,mathrsfs}
\usepackage{amsfonts}
\usepackage{tikz}
\usepackage{xcolor}
\usepackage{pgfplots} 
\usepackage{cite}
\usetikzlibrary{positioning}
\tikzset{main node/.style={circle,fill=blue!20,draw,minimum size=1cm,inner sep=0pt},
            }
\usepackage{amssymb}
\usepackage{amsfonts}
\usepackage{afterpage}
\usepackage{graphicx}
\graphicspath{ {c:/user/images/} }
\usepackage{float}
\newtheorem{teo}{Theorem}
\newtheorem{defi}{Definition}
\newtheorem{prop}{Proposition}

\newtheorem{rem}{Remark}
\usepackage{todonotes}
    \setuptodonotes{color=blue!30}

\newcommand{\absolute}[1]{\left\lvert#1\right\rvert}

\DeclareMathOperator{\Vol}{Vol}
\DeclareMathOperator{\Spec}{Spec}

\DeclareMathOperator{\supp}{supp}

\title{Time-Varying Energy Landscapes and Temperature paths: Dynamic Transition Rates in locally Ultrametric Complex Systems .}
\author{
\'Angel Mor\'an Ledezma \\
angel.ledezma@kit.edu 
\\
Institute of Photogrammetry and Remote Sensing
\\
Karlsruhe Institute of Technology
\\
Englerstr.\ 7
\\
76131 Karlsruhe
\\
Germany
}

\date{}

\begin{document}

\maketitle

\begin{abstract}
In this work, we study the dynamics of complex systems with time-dependent transition rates, focusing on $p$-adic analysis in modeling such systems. Starting from the master equation that governs the stochastic dynamics of a system with a large number of interacting components, we generalize it by $p$-adically parametrizing the metabasins to account for states that are organized in a fractal and hierarchical manner within the energy landscape. This leads to a not necessarily time homogeneous Markov process described by a time-dependent operator acting on an ultrametric space. We prove well-posedness of the initial value problem and analyze the stochastic nature of the master equation with time-dependent transition-operator. We demonstrate how ultrametricity simplifies the description of intra-metabasin dynamics without increasing computational complexity. We apply our theoretical framework to two scenarios: glass relaxation under rapid cooling and protein folding dynamics influenced by temperature variations. In the glass relaxation model, we observe anomalous relaxation behavior where the dynamics slow down during cooling, with lasting effects depending on how drastic the temperature drop is. In the protein folding model, we incorporate temperature-dependent transition rates to simulate folding and unfolding processes across the melting temperature. Our results capture a "whiplash" effect: from an unfolded state, the system folds and then returns to an unfolded state (which may differ from the initial one) in response to temperature changes. This study demonstrates the effectiveness of $p$-adic parametrization and ultrametric analysis in modeling complex systems with dynamic transition rate, providing analytical solutions that improve our understanding of relaxation processes in material and biological systems. 

\end{abstract}

\tableofcontents

\section{Introduction}

In the physics of complex systems, the so called master equation is a fundamental tool for describing the stochastic dynamics of systems with a large number of interacting components. \cite{Bowman2014Markov,Husic2018Markov,Peliti2021Stochastic,Mauro2021Materials}. These systems are constituted by a huge amount of possible configurations, each corresponding to a specific state of the system. The master equation provides a mathematical framework to model the time evolution of the probability distribution over these states, assuming the system randomly jumps from one state to another.

A central concept in this context is the energy landscape; a multidimensional surface representing the potential energy of a system as a function of its configurations. Each point on this landscape corresponds to a specific configuration, and the local minima of this surface represent stable or metastable states where the system tends to stay due to lower potential energy compared to neighboring configurations. \cite{ABZ2014,ABKO2002,WPG2013,Frauenfelder2010Physics,Frauenfelder1991Energy,Khrennikov2021Ultrametric,Mauro2012Minimalist,Peliti2021Stochastic,Mauro2021Materials}

The master equation describe a stochastic process (a continuous time Markov chain (CTMC); a time homogeneous Markov process), where the local minima of the energy surface are mapped to states of this Markov chain, and the transitions between them are governed by transition rates determined by the energy barriers separating the configurations. The probabilities in this process can be interpreted as relaxation processes where systems evolve toward equilibrium. \cite{ZunigaNetworks2,Mauro2021Materials,Peliti2021Stochastic,Nolting2005Protein}.The transition rates often depend on factors like temperature and are modeled using Arrhenius-type expressions, reflecting the likelihood of the system overcoming energy barriers. This allow us to model the dynamics in systems ranging from protein folding to glass relaxation and chemical kinetics. 

For example, in protein folding, each state represents a distinct three-dimensional conformation of a protein molecule. The master equation models how the protein explores its energy landscape to find the lowest energy conformation, navigating through various metastable states \cite{Frauenfelder2010Physics,Frauenfelder2003Myoglobin,Nolting2005Protein}. In glass relaxation, the master equation describes how the system explores a rugged energy landscape with many metastable states. \cite{Mauro2021Materials,Mauro2007Metabasin}.

While the energy landscape approach offers valuable insights into the multiple timescales and governing physics of glass transitions and and relaxation behaviors including protein folding, it can be computationally intensive due to the huge number of local minima or states involved. As J. C. Mauro et. al. point out \cite{Mauro2012Minimalist}, it is crucial to develop simplified or analytically solvable models that capture the essential features of these complex processes. Moreover, recognizing the limitations and predictions of Markov state models is important in order to have a deeper understanding on complex systems. One important assumption in the modeling of such processes is the time-independence of the transition rates. Nevertheless, the transition rates on certain glass relaxation models depend on a 'temperature path', implying that temperature is not constant overtime, making this transition rate time-dependent \cite{Mauro2021Materials,Guo2012Temperature}. Furthermore, like all chemical reactions protein folding depends on its environment; hence, temperature and the energy landscape itself could be not constant over time. This variability may arise from artificial manipulations, where temperature changes are induced by artificial modulation  (see for example\cite{Guo2012Temperature}) , or from intracellular thermal interactions that cause energy barriers and temperatures to fluctuate, moreover, due to their dynamic nature, proteins spontaneously unfold and refold many times in vivo as shown in \cite{WPG2013}, as also capture by our simulations in chapter 3. Therefore, understanding state models with variable transition rates in realistic contexts---such as in glasses and proteins---becomes important for capturing their real-world dynamics. This work aims to give some insights in this direction. 

Dealing with a large number of transition rates poses significant challenges for the solvability of the problem, even when transitions are constant. Introducing time-dependent transitions further complicates the scenario. To address this complexity, we employ the approach inspired by W. Zuñiga, where interactions within metabasins are modeled using \( p \)-adic parametrization. This concept, used by Zuniga in \cite{ZunigaNetworks2}, has proven to be a powerful tool, allowing for the representation of an arbitrarily large number of states within each metabasin. Moreover,  in \cite{ZunigaNetworks2}, W. Zuñiga was able to show how this model allow to rigorously explains how, under some non-equilibrium conditions, a complex system can reach absorbing states by imposing two different transition functions.

Our motivation for utilizing \( p \)-adic analysis comes not only from its computational advantages but also from its widespread use in the literature as a tool to model hierarchical complex systems, see \cite{ABZ2014,XK2018b,ABKO2002,DragovichBA2010,ABK1999,DXKM2021,Avetisov2002Padic,Khrennikov2021Ultrametric}, and the references there in. It has been proposed as a valuable tool for analyzing energy landscapes, were recent studies such as \cite{Charbonneau2014Fractal,jammingultra} confirming the fractal and hierarchical  behavior of the energy landscape in structural glasses , and well known studies showing the ultrametric nature of protein folding energy landscapes \cite{Frauenfelder2010Physics} serves also as main motivation for the ultrametric methods in this work. In this article, we show how the \( p \)-adic parametrization offers a novel technique for modeling complex systems with time dependent transition rates, since the ultrametric parametrization of this states allow us to give an analytical solution to the contribution of these states to the overall the dynamics. That is, we are able to describe a large number of states without increasing computational complexity. This approach is particularly effective in describing the dynamics of a multitude of states inside a metabasin.

This article aims to explore the dynamics of a CTMC with time dependent transition rates, highlighting how \( p \)-adic analysis emerges as a powerful and natural tool in studying these processes. Now we proceed to describe our results. 

In chapter $2$ we introduce the general master equation for a discrete system space and what we understand by  $p$-adic parametrization this leads to a generalization of the classical CTMC differential equaiton: 

\[
\frac{d}{dt} f_a(t) = \sum_b [ w_{ab} f_b(t) - w_{ba} f_a(t) ] \rightarrow \frac{d}{dt} f(x,t) = \int_{\mathbb{Z}_p} w(|x - y|)[ f(y,t) - f(x,t) ]  dy
\]
Where $w_{ab}$ represent the transition rate probability between states $a$ and $b$, $w(|\cdot|_p)$ is a radial function on $\mathbb{Z}_p$, the $p$-adic integers, which can be seen as a infinitesimal cluster of states hierarchically organized following a fractal, or self-similar, structure. We then consider a finite set of this clusters, given by translations of $\mathbb{Z}_p$ these are the metabasins of the model. The transitions inside each metabasin are controlled by a radial function, while the transitions between the metabasins $I+\mathbb{Z}_p$ are controlled by a transition rate matrix. The dynamic of such a system is introduced and briefly studied in Section 2.2. After this we allow time dependence on the transitions by introducing the time dependent operator  
\[\mathbf{W}(t)f(x)=\int_{K_N}[w(x,y,t)f(y)-w(y,x,t)f(x)]dy.\]
where $K_N$ is the disjoint union of the metabasins and $w(x,y,t)$ is the transition rate function. The master equation attached to this operator studied in Chapter 2-Section 2.3, give rise to a different stochastic process (a strong Markov process, not necessarily time-homogeneous). Theorem 1 is our first main result, here we show that the initial value problem attached to this process is well-possed. Furthermore, in Theorem 2 we describe the stochastic nature of the master equation attached to the time dependent operator $\mathbf{W}(t)$. Section 2.4 is central on the following resutls since here we expose how ultrametricity ($p$-adic parametrization of the metabasins) leads to an easy description of its contribution on the overall dynamics, this is achieved by using the Trotter-Kato Theorem for semigroups. 

In Chapter 3 we introduce a two metabasin model; where we denote by $B_u=\mathbb{Z}_p$ the higher energy metabasin  (which in the case of our protein example correspond to the unfolded basin, with several unfolded-degenerate states, and by  $B_f$ the lower energy one, corresponding to the folded basin in the protein example, with several folded-degenerate states). This model can be considered as a time-dependent generalization of the one given by the minimalist energy introduced in \cite{Mauro2012Minimalist}. We are interested in understand the behavior of the characteristic relaxation of a sub basin of $B_U$, $B_{r_0}\subset B_U$. Here, the characteristic relaxation of this region, $S(t)$, will be understood as the evolution of population (or occupation probability) in the domain of the initial distribution, in this case $B_{r_0}$. In order to describe the function $S(t)$ we use Trotter Kato to describe the contribution of the transitions between meta-basins (or in the terminology of J. C. Mauro in \cite{Mauro2012Minimalist} on glass-relaxation the $\alpha$ transition, and in the context of protein folding, the transition to the unfolded to folded metabasins), while the contribution to $S(t)$ given by the intra-metabasin dynamic $p$-adically parametrized (which correspond to the $\beta$ transitions or the dynamic between unfolded states) is directly computed by the previous analysis on chapter 2. The characteristic relaxation $S(t)$ has the following form.
\[
S(t) = \underbrace{p^{-r_0} p_1(t)}_{\text{metabasin transition contribution}} + \underbrace{p^{-r_0} \sum_{\substack{\operatorname{Supp} \, \psi_{r,j,n} \\ \not\subset B_U}} |C_{r,j,n}|^2 e^{-\int_{0}^{t} \gamma_{r,U}(\tau) \, d\tau}}_{\text{intra-metabasin (}\!p\text{-adic) contribution}}.
\]
We then apply our results to two scenarios. In the first scenario, we use similar data used in \cite{Mauro2012Minimalist} to modeled the transitions of a glass-relaxation phenomena, while allowing the temperature to drop in a fast way with the purpose of simulating the effect of rapid cooling at different temperatures. We can observe in section 3.2 how an anomalous relaxation take place; the dynamic is slowed down during the cooling and the effect depend on the magnitud of the temperature drop. Moreover the slowdown in the dynamics has a long lasting effect, even when the temperature ceases to drop, with the effect depending on how drastic the temperature change has been over time.  In the second scenario we use the temperature dependence on the transition rates of protein folding studied in \cite{Guo2012Temperature} modeled by the equations: 
\[
\ln\left(\frac{k_0}{k_f(T)}\right) = \frac{1}{R T} \left( \Delta H_f - T \Delta S_f + \Delta C_p^f \left( T - T_m + T \ln\left( \frac{T_m}{T} \right) \right) \right),
\]

\[
\ln\left(\frac{k_0}{k_u(T)}\right) = \frac{1}{R T} \left( \Delta H_u - T \Delta S_u + \Delta C_p^u \left( T - T_m + T \ln\left( \frac{T_m}{T} \right) \right) \right),
\]
where the thermodynamical parameters associated with the in vitro case are also given in this source. In this case we observe a highly different dynamic with the two-basin model: The temperature dependence transitions allow us to describe in time the rate-transitions, which vary in a way that allow the system to go from a temperature below the melting point and above the melting point. This is reflected in the behavior of $S(t)$, where the system go from an unfolded state with probability $1$ at time zero, to a "whiplash" effect, where the system travel to the folded state with high probability, and then return (after crossing the melting temperature) to the unfolded metabasin. Moreover, even though the system return with high probability to an unfolded state, does not necessarily return to the original state $B_{r_0}$, this state can now be occupied with a lower probability (corresponding to the volume of the sub-region) which tell us how the system can occupy any other unfolded state in the return. 

\section{Time dependent transition rates and $p$-adic transition networks.}


\subsection{The \(p\)-adic numbers and trees.}
Let $\mathbb{Q}$ be the set of rational numbers equipped with the $p$-adic norm. That is, for a given prime number $p$ and $\frac{a}{b}\in \mathbb{Q}\setminus\{0\}$, we define 
\[
|x|_p := p^{-\nu(x)}
\]
where  $\frac{a}{b}=p^\nu \frac{a'}{b'}$ for some suitable $\nu\in\mathbb{Z}$ and
some co-prime integers $a'$ and $b'$. By setting $|0|_p := 0$, the function $|\cdot|_p:\mathbb{Q}_p\rightarrow \mathbb{R}_{\geq0}$ define a norm on $\mathbb{Q}$.  The completion of $\mathbb{Q}$ via $|\cdot|_p$ is called the field of $p$-adic numbers, and is denoted by $\mathbb{Q}_p$. The space $(\mathbb{Q}_p,|\cdot|_p) $ is an ultrmatric space, that is the ultrametric inequality holds: $|x+y|_p\leq \max \{|x|_p,|y|_p\} $. The space $\mathbb{Q}_p$ has attached a Haar measure denoted by $dx$. Each $p$-adic number $x\in \mathbb{Q}_p$ has can be represented as a convergent Laurent series of the form
\[
x=\sum_{k=\nu}^\infty x_k p^k.
\]
where $x_k\in\{0,\dots,p-1\}$. We denote the unit ball of this metric space by $\mathbb{Z}_p:=\{x\in \mathbb{Q}_p: |x|_p\leq 1 \}.$  A function $w:\mathbb{Q}_p\rightarrow \mathbb{R}$ is called radial if $w(x)=w(|x|_p)$.  \newline

Let $\mathbb{Z}_p/p^n\mathbb{Z}_p:=G_n$ be the finite group consisting on the elements of the form
\[G_n\ni x=a_0+a_1p+...+a_{n-1}p^{n-1},\]
where $a_i\in \{0,...,p-1\}$. Therefore we have the following decomposition. 
\[\mathbb{Z}_p=\bigsqcup_{a\in G_n}a+p^n\mathbb{Z}_p.\]

\begin{figure}[H]
    \centering
\begin{tikzpicture}[
  grow=down,
  level distance=1.5cm,
  level 1/.style={sibling distance=4cm},
  level 2/.style={sibling distance=2cm},
  level 3/.style={sibling distance=1cm},
  every node/.style={draw, rectangle, minimum width=1cm, inner sep=2mm, align=center}
  ]
  \node {} 
    child {node {0}
      child {node {00}
        child {node {000}}
        child {node {001}}
      }
      child {node {01}
        child {node {010}}
        child {node {011}}
      }
    }
    child {node {1}
      child {node {10}
        child {node {100}}
        child {node {101}}
      }
      child {node {11}
        child {node {110}}
        child {node {111}}
      }
    };

  \node[anchor=east] at (-4,0) {Level 0: radius $2^0=1$};
  \node[anchor=east] at (-4,-1.5) {Level 1: radius $2^{-1}=\frac{1}{2}$};
  \node[anchor=east] at (-4,-3) {Level 2: radius $2^{-2}=\frac{1}{4}$};
\end{tikzpicture}
    \caption{Tree representation of $G_3$ for $p=2$}
    \label{fig:tree}
\end{figure}

The elements of the group $G_n$ are in one-to-one correspondence with the leafs of a tree. In the tree, the length of the common initial path from the root indicates how many leading digits the two elements share; thus, the node where their paths first diverge determines the $p$-adic distance between them. In particular, two leaves that share a common path of length $k$ are at distance $p^{-k}$ apart, as shown in figure \ref{fig:tree}.

Every node at a given level represents a $p$-adic ball whose radius is determined by the number of fixed digits. The tree structure thereby encapsulates both the algebraic representation of $G_n$ and the topology given by the $p$-adic metric, where the closeness of any two elements is reflected in the depth of their most recent common ancestor in the tree. \\

\subsection{Ultrametric dynamics and fractal energy landscapes}

The dynamics of a physical system subject to random interactions with a heat reservoir can be modeled by a master equation \cite{Mauro2021Materials}. This is a central idea in stochastic thermodynamics and complex systems. This approach involves the description of a random walk on an energy landscape containing a very large number of local minima, which constitute the stable conformations (or states) of the complex system. The master equation becomes a useful technique for modeling the fluctuations and relaxations of such systems, since these phenomena correspond in this model, to jumps between states. 

\begin{figure}[H]
    \centering
 \begin{tikzpicture}[scale=0.9, domain=-1:9]
  \coordinate (start) at (-1,2.5); 
  \coordinate (m1) at (0,0);
  \coordinate (m2) at (1,0);
  \coordinate (m3) at (2,0);
  \coordinate (m4) at (3,0);
  \coordinate (m5) at (5,0);
  \coordinate (m6) at (6,0);
  \coordinate (m7) at (7,0);
  \coordinate (m8) at (8,0);
  \coordinate (end) at (9,2.5); 

  \foreach \m in {m1,m2,m3,m4,m5,m6,m7,m8}{
    \filldraw (\m) circle (2pt);
  }

  \draw[thick, smooth]
    (start)
    .. controls (-0.5,1.5) .. (m1)
    .. controls (0.5,0.5) .. (m2)
    .. controls (1.5,1) .. (m3)
    .. controls (2.5,0.5) .. (m4)
    .. controls (4,2) .. (m5)
    .. controls (5.5,0.5) .. (m6)
    .. controls (6.5,1) .. (m7)
    .. controls (7.5,0.5) .. (m8)
    .. controls (8.5,1.5) .. (end);

  \foreach \i/\m in {1/m1,2/m2,3/m3,4/m4,5/m5,6/m6,7/m7,8/m8}{
    \node[below] at (\m) {$m_{\i}$};
  }

  \node[above] at (-0.5,1.5) {};
  \node[above] at (0.5,0.5) {$\color{green} w_3$};
  \node[above] at (1.5,1) {$\color{blue} w_2$};
  \node[above] at (2.5,0.5) {$\color{green} w_3$};
  \node[above] at (4,2) {$\color{red} w_1$};
  \node[above] at (5.5,0.5) {$\color{green} w_3$};
  \node[above] at (6.5,1) {$\color{blue} w_2$};
  \node[above] at (7.5,0.5) {$\color{green} w_3$};
  \node[above] at (8.5,1.5) {};

  \draw[->] (-1.5,0) -- (-1.5,2.8) node[left] {Energy};

  \draw[->] (-1.5,-0.5) -- (9.5,-0.5) node[below] {Configuration space};

   \draw[dashed] (-1,2.5) -- (-1,0);
   \draw[dashed] (9,2.5) -- (9,0);

\end{tikzpicture}
    \caption{Hierarchical (rugged) energy landscape where each minima is in one-to-one correspondence with an element of $G_3$ for $p=2$ }
    \label{fig:localenergy}
\end{figure}
A master equation is defined by the jump rates $w_{a,b}(t)>0$ from a discrete state $b$ to $a$. The system of equations has the form 
\[\dfrac{d}{dt} f_a(t)=\sum_{b}[w_{a,b}(t)f_b(t)-w_{b,a}(t)f_a(t)],\]
where $f_a(t)$ is the probability distribution or population in the local minimum $a$. We say that the jump rate $w_{a,b}$ admits a $p$-adic parametrization if $w_{a,b}=w(|a-b|_p,t)$ where $a,b\in G_n$  for $n\geq 0$, see Figure \ref{fig:localenergy}. If $n\rightarrow \infty$, that is, the number of states or local minima tends to infinity,  we can obtain the infinitesimal $p$-adic system described by the equation
\begin{equation}
    \label{balleq}
\frac{d}{dt}f(x,t)=\int_{\mathbb{Z}_p}w(|x-y|,t)(f(y,t)-f(x,t))dy.
\end{equation}
When the function $w(|\cdot|_p,t)= w(|\cdot|_p)$ does not depend on time, the attached master equation describes the diffusion on a $p$-adic ball studied, for example in, \cite{ZunigaNetworks2}.  The attached energy landscape is a fractal, hierarchically organized as shown in figure \ref{fig:localenergy}, but this time we have an infinite sequence of energy barriers; each radial value of $w$ depend on the corresponding height of the barrier, as schematically shown in figure \ref{fig:fractalenergy}.

\begin{figure}[H]
    \centering
    \includegraphics[width=0.6\linewidth]{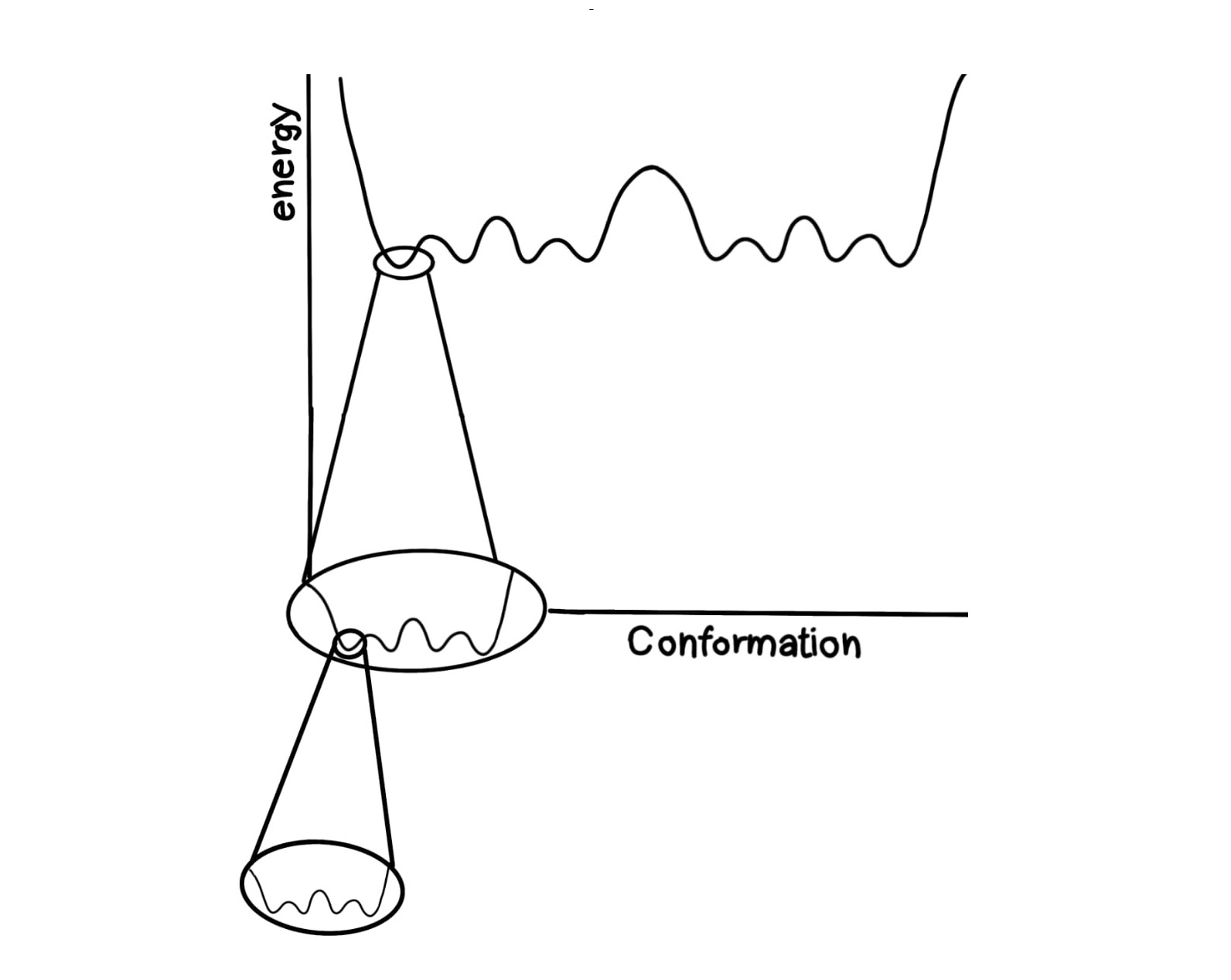 } \caption{Ultrametric energy landscape, where the infinite number of energy barriers follow a self-similar and hierarchical organization in sub-basins, all minima have the same energy level (degenerate).}
    \label{fig:fractalenergy}
\end{figure}

Now, we are interested in a more general situation. Consider a disjoint collection of $p$-adic balls of radius 1
\[
B(I)=\{x\in \mathbb{Q}_p: |x-I|\leq 1\},
\]
where $I$ belongs to a finite set $V$ with cardinality $|V|=N$ . Let 
\[
K_N =\bigsqcup_{I\in V}B(I). 
\]
Each ball represents a cluster of states or minima, or as described in, cf.\ e.g.\ \cite{ABKO2002}, these balls represent separated basins of an energy landscape. Each state inside these basins represent a minima which is hierarchically organized.  \\

Like all chemical reactions, protein folding is dependent on its environment, and in the case of glass relaxation, the process may depend on a variable temperature. In the literature, including that of $p$-adic  mathematical physics, the energy landscape is usually considered to be constant in time. Nonetheless, there are several situations where the transition rates are time-dependent, e.g.\  if the system is subjected to artificial modulation  \cite{templandscape2}, or as  described in 
\cite{WPG2013},
sometimes the transition rates on the folding process of a protein are time-dependent for natural reasons, such as the cell cycle itself. \\

Therefore, we propose the study of the dynamic generated by a master equation of the form 
\[
\frac{d}{dt}f(x,t)=\int_{K_N}[w(x,y,t)f(y,t)-w(y,x,t)f(x)]dy,
\]
where the non-necessarily symmetric function $w:K_N \times  K_N\times (0,\infty) \rightarrow (0,\infty)$ satisfies the following definition.  

\begin{defi}
    A time-dependent $p$-adic transition function $w:K_N\times K_N \rightarrow (0,\infty)$ is a function of the form 
\[w(x,y,t)=\sum_{I,J\in V} w_{I,J}(x,y,t)\Omega(p^N|x-I|_p)\Omega(p^N|y-J|_p)\]
where
\[w_{I,I}(x,y,t)=w(|x-y|_p,t)\]
are bounded radial functions and 
\[w_{I,J}(x,y,t)=w_{I,J}(t)
\]
for $I\neq J$. When $w(x,y,t)=w(x,y)$ we call this function an autonomous $p$-adic transition function. 
\end{defi}

The function $w(x,y,t)$ represents the transition rate peer unit of time,  which usually follows an Arrhenius type equation. Therefore, the above definition implies that, the transition rate between two points $x,y\in B(I)$, satisfies $w(x,y,t)=w(|x-y|_p,t)$, i.e. the dynamic depend on energy barriers which are hierarchically organized (Figure \ref{fig:fractalenergy}). Meanwhile, for $I\neq J$, the transition rate function satisfies $w(x,y,t)=w_{IJ}(t)$, that is, the transition rates between two basins  at a given time $t>0$ is constant  and independent of the points $x\in B(I)$, $y\in B(J)$.  Note that in general $w_{IJ}(t_0)\neq w_{JI}(t_0)$, hence this dynamic allow us to consider non-degenerate landscapes in contrast with past ultrametric models such as \cite{ABKO2002}. Hence we are working with a \textit{locally ultrametric energy landscape}. See figure \ref{fig:locallyultrametric} . 

\begin{figure}[H]
    \centering
    \includegraphics[width=0.8\linewidth]{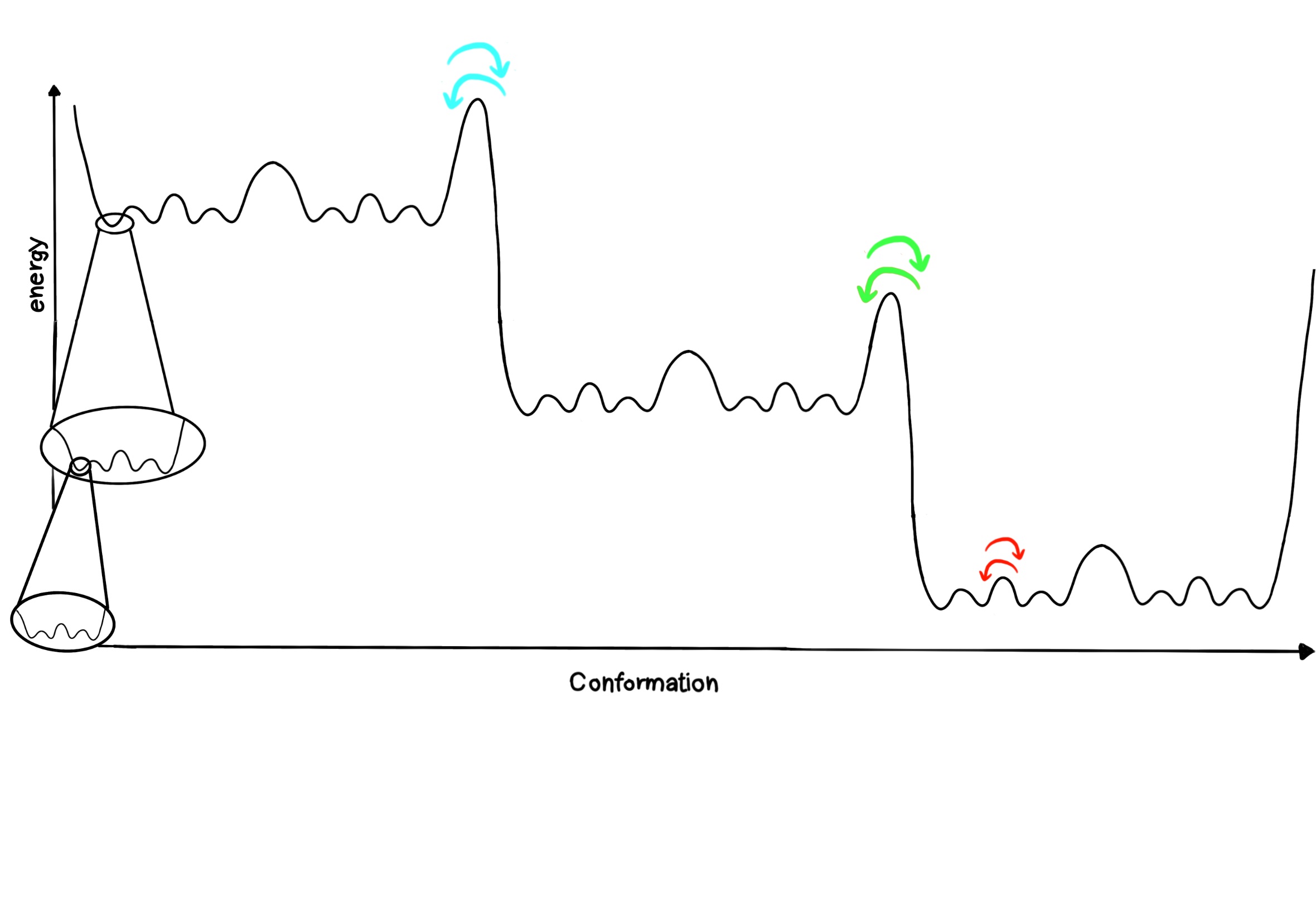}
    \caption{Locally-ultramemtric landscape: The blue and green transitions are controlled by a constant $w_{I,J}$, meanwhile inside each meta-basin (red-arrows) the transitions follow a fractal and hierarchical organization controlled by a radial function $w_{I,I}(|\cdot|_p).$}
    \label{fig:locallyultrametric}
\end{figure}

In order to study the dynamics of such a system, we introduce some important functional spaces. Let $C(K_N)$ denote the space of complex valued continuous functions on $K_N$. The space of absolute integrable functions on $K_N$ is denoted by $L^1(K_N)$ and $L^2(K_N)$ denotes the usual Hilbert space of functions over the space $K_N$. In order to simplify the computations we assume that $B_r=\mathbb{Z}_p$ and $V\subset\mathbb{Q}_p /\mathbb{Z}_p$. An orthonormal basis of the space $L^2(\mathbb{Q}_p)$ is given by the set of functions $\{\psi_{rjn}\}$ defined as 
\begin{equation}\label{KozyrevWavelet}
    \psi_{rjn}(x)=p^{\frac{r}{2}}\chi_p(p^{-1-r}jx)\Omega(|p^{-r}x-n|_p),
\end{equation}
where $r\in \mathbb{Z}$, $j\in \{1,...,p-1\}$, and $n\in \mathbb{Q}_p /\mathbb{Z}_p$ which are called \textit{the Kozyrev wavelet functions}. In particular, it can be proved that for any $f\in L^2(\mathbb{Z}_p)$, the following expansion holds,
$$f(x)=p^{1/2}\Omega(|x|_p)f_0+\sum_{rjn} C_{rjn} \psi_{rjn}(x),$$
where 
$$f_0=\int_{\mathbb{Z}_p}f(x)dx,$$
and $C_{rjn}\in \mathbb{C}$ and $\psi_{rjn}$ are the Kozyrev wavelets supported on $\mathbb{Z}_p$. Moreover the above expansion converges uniformly to $f$, for any $f\in C(\mathbb{Z}_p)$. Thus these functions form a basis of the Hilbert space $L^2(\mathbb{Z}_p)$. This basis can be generalized to a basis of $L^2(K)$, for any compact set $K\subset \mathbb{Q}_p$. Let $K:=\bigsqcup_{I\in V} I+\mathbb{Z}_p$ then the following decomposition as a direct sum of Hilbert spaces holds,
\[L^2(K)=\bigoplus_{I\in V} L^2(I +\mathbb{Z}_p),\]
then for each $f\in L^2(K)$, we have
\begin{equation}\label{expansionbasis}
    f(x)=\sum_{i=1}^{N} \left(\left(\int_{a_i+\mathbb{Z}_p}f(|x|_p)dx\right) p^{1/2}\Omega(\absolute{x-a_i})+\sum_{rjn}C_{rjn}^{(i)}\psi_{rjn}^{I}(x)\right),
\end{equation}
where $\psi_{rjn}^{I}$ denotes a Kozyrev wavelet function supported on $I+\mathbb{Z}_p$. Since 
\[ \int_{I+\mathbb{Z}_p} \psi_{rjn}^{I}(x)dx=0, \]
the latter implies the following decomposition,
\[ L^2(K)=\mathbb{C}^{|V|}\oplus \mathcal{L}_0(K), \]
where  $\mathcal{L}_0(K),$ is space of functions with mean zero on $K$, which coincides by the space generated by the functions $\psi_{rjn}^{I}(x)$ as a Hilbert space. Here we have identified the finite dimensional complex vector space $span_{\mathbb{C}}\{p^{1/2}\Omega(\absolute{x-I}_p)\}$ with $\mathbb{C}^{|V|}$.

\subsection{Time-dependent $p$-adic operators and Markov processes.}\label{nonAutoEq}
Let $w:K_N \times K_N \times [0,\infty)\rightarrow (0,\infty)$ be a time-dependent $p$-adic transition function. Then we define the time-dependent operator $\mathbf{W}$ as the function $t\mapsto \mathbf{W}(t)$. Where 
\[\mathbf{W}(t)f(x)=\int_{K_N}[w(x,y,t)f(y)-w(y,x,t)f(x)]dy.\]
Define the degree function as
\[
\gamma(t)(x)=\sum_{I\in V(\mathcal{G})}\gamma_I(t)\Omega_{I,N}(x)
\]
with 
$\gamma_{I}(t)=\int_{B_N(I)}w_{I,I}(|x|_p,t)dx +\sum_{J\in V(\mathcal{G})}w_{IJ}(t).$ Then, we can rewrite the operator as follows 
\[\mathbf{W}(t)f(x)=\int_{K}w(x,y,t)f(y)dy-\gamma(t)(x)f(x).\]
Now we will study the Cauchy problem associated with the time dependent operator $\mathbf{W}(t)$, given by:
\begin{equation}\label{nAsyst}
    \begin{cases}
      \frac{\partial u}{\partial t} (x,t)=\mathbf{W}(t)u(x,t)\\
      u(x,s)=u_s(x)\in C(K,\mathbb{C}),
    \end{cases}
\end{equation}
this initial value problem has attached an stochastic process which is described in Theorem \ref{Markov}. In order to state the following results we need some definitions. 

\begin{defi}
A continuous function $u:[s,\infty)\rightarrow C(K,\mathbb{C})$ is called a (strict) \textbf{solution} of (\ref{nAsyst}), if $u\in C^1([s,\infty),C(K,\mathbb{C}))$, $u(t)\in D(\textbf{W}(t))$ for all $t\geq s$, $u(s)=x$, and $\frac{\partial u}{\partial t}=\textbf{W}(t)u$ for $t\geq 0$.
\end{defi}
We say the problem \ref{nAsyst} is \textit{well-posed} when for any initial condition and any time $s\geq 0$ there exists a unique solution as in the above definition.  If the time-dependent operator $t\rightarrow \mathbf{W}(t)$ is strongly continuous (that is, $t\mapsto A(t)u$ is continuous for each $u\in X$), then the system \ref{nAsyst} is well-possed as stated in the next proposition. 

\begin{prop}\label{fattoriniTheorem}
Let $X$ a Banach space and for every $t>0$ let $A(t)$ be a bounded linear operator on $X$. If the function $t\mapsto A(t)$ is strongly continuous for $0\le t<T$
 then for every initial condition $x\in X$, the attached initial value problem is well-posed.
\end{prop}

\begin{proof}
\cite[Thm.\ 7.1.1]{Fattorini1983}.
\end{proof}

\begin{teo}\label{stronglyContinuous}
Let $w(x,y,t)$  be a non-autonomous $p$-adic transition function if the functions $t\mapsto w_{I,J}(t)$, and $(x,t)\mapsto w_{I,I}(|x|_p,t) $ are uniformly continuous for each $I,J\in V$, then problem \ref{nAsyst} is well-possed.
\end{teo}
\textit{Proof.} 
In virtue of proposition \ref{fattoriniTheorem} its enogh to prove that the time-dependent operator  $t\mapsto \mathbf{W}(t)$ is strongly continuous. Let $u\in L^2(K_N,\mathbb{C})$ , we want to estimate the following norm,
\[||\mathbf{W}(t)u-\mathbf{W}(s)u||_{L^2}\]
For $\varepsilon>0$ and every $x\in K$, there is a positive real number $\delta>0$ such that 
\[
|t-s|<\delta \implies|w_{IJ}(t)-w_{IJ}(s)|<\frac{\varepsilon}{n^2},
\]
and
\[
|t-s|<\delta \implies|w_{II}(|x|_p,t)-w_{II}(|x|_p,s)|<\frac{\varepsilon}{n^2},
\]
Hence
\begin{equation*}
    \begin{split}
        |w(x,y,t)-w(x,y,s)|\leq&\ |w_{II}(|x|_p,t)-w_{II}(|x|_p,s)|+\sum_{I\neq J}|w_{IJ}(t)-w_{IJ}(t)|\\
        &\leq \varepsilon.
    \end{split}
\end{equation*}
Therefore
\begin{equation*}
    \begin{split}
        \int_{K_N} \left(\int_{K_N}|(x,y,t)-A(x,y,s)||u(y)|dy\right)^2dx &\leq \varepsilon^2 \int_{K_N}\left(\int_{K_N}|u(y)|dy\right)^2dx\\
        &\leq \varepsilon^2\Vol(K_N)||u||_{L^2}^2,
    \end{split}
\end{equation*}
where the last inequality is due to Hölder's inequality. 
On the other hand, 
\begin{align*}
|\gamma(t)(x)-\gamma(s)(x)|&=\left|\sum_{I\in
V(\mathcal{I})}(\gamma_I(t)-\gamma_I(s))\Omega_{I,N}(x)\right|
\\
&\leq\sum_{I}\left|w_{II}(|x|_p,t)-w_{II}(|x|_p,s)+\sum_{J\neq I}w_{IJ}(t)-w_{IJ}(s)\right|
\\
&\leq \varepsilon
\end{align*}
Hence, 
\begin{equation*}
    \begin{split}
         \int_{K}|\gamma(t)(x)-\gamma(t)(x)|^2|u(y)|^2dx
        \leq \varepsilon^2 ||u||_{L^2}^2
    \end{split}
\end{equation*}
Finally, we have 
\begin{align*}
||\mathbf{W}(t)u-\mathbf{W}(s)u||_{L^2}
\leq \varepsilon||u||_{L^2}\left(\Vol(K_N)^\frac12+1\right)  
\end{align*}
 This ends the proof. \qed \newline
 
\begin{rem}
  This result and the following ones generalize the results presented in   \cite{nonAutonomousDiffusion}, here the transition rates inside each ball (or each metabasin) are not zero (compare with \cite{nonAutonomousDiffusion}), that is, here we take in account variable rates which model the transitions between states inside each basin. Therefore, the dynamic is fundamentally different, this results are also different from the Zúñiga model presented in \cite{ZunigaNetworks2}. On the other hand, our previous work \cite{nonAutonomousDiffusion}, can be considered the time-dependent generalization of \cite{ZunigaNetworks}, which may be useful to study Turing patterns on  time-changing graphs.
\end{rem}

In particular the proof of Theorem \ref{stronglyContinuous} shows that $t\mapsto \mathbf{W}(t)$ is continuous in the uniform operator topology. 

As stated before, the master equation \ref{nAsyst}, has attached an stochastic process. This process is a strong Markov process (in fact a Hunt process)  (for definitions we refer the reader to \cite{vanCasteren2011}). The solution of the master equation \ref{nAsyst} can be given in terms of an evolution family, $(P(t,s))_{t\geq s}$ of linear operators \cite{Schnaubelt2006}. This family is defined by 
\begin{equation}
    \label{semigroupsol}
P(t,s)u(x)=u(x,t)
\end{equation}
where $u(x,t)$ is the solution of \ref{nAsyst} for the initial condition $u(x,s)=u(x)$.  \newline
There is a natural correspondence between Feller evolution families and transition probability functions of non-homogeneous Markov processes \cite[Thm.\ 2.9]{vanCasteren2011}. For the sake of completeness we review the definition of a Feller Evolution. 

\begin{defi}
\label{FellerEvolution}
A family $\{P(s,t):0\leq s \leq t \leq T\}$ of operators defined on $L^{\infty}(K)$ is called a \textbf{Feller Evolution} on $C_b(K)$ if it possesses the following properties: 
\begin{enumerate}
    \item It leaves $C_b(K)$ invariant: $P(s,t)C_b(K)\subset C_b(K)$ for $0\leq s\leq t \leq T$;
    \item It is an evolution: $P(\tau,t)=P(\tau,s)P(s,\tau)$ for all $\tau, s, t$ for which $0\leq \tau \leq s \leq t$ and $P(t,t)=I$, $t\in [0,T]$;
    \item If $0\leq f \leq 1$, $f\in C_b(K)$, then $0\leq P(s,t)f \leq 1$, for $0\leq s \leq t\leq T$. 
    \item If the function $(s,t,x)\mapsto P(s,t)f(x)$ is continuous on the space $\Lambda:=\{(s,t,x)\in [0,T]\times[0,T]\times E:s\leq t\}$. 
\end{enumerate}
\end{defi}
As expected, the evolution family \ref{semigroupsol} is a Feller evolution. For this, we need the next result which proof can be found in \cite{nonAutonomousDiffusion}. 

\begin{prop}\label{contractionSG}
Let $\{A(t)\}_{t\geq 0}$ be a set of bounded operators on $C(X)$, where $X$ is a Banach space. Suppose that $t\mapsto A(t)$ is continuous in the uniform norm topology and that each $A(t)$ generates a strongly continuous, positive, contraction semi-group. Then the Cauchy problem 
\begin{equation}\label{nACP_abstract}
    \begin{cases}
      \frac{\partial u}{\partial t} (t)=A(t)u(t)\\
      u(s)=x\in X
    \end{cases}
\end{equation}
is well-posed and its solution evolution family $P(t,s)$, given generates a Feller Evolution.
\end{prop}
We are now ready to state the main result of this section.
\begin{teo}\label{Markov}
Let $w(x,y,t)$ be a time-dependent $p$-adic transition function such that the functions $t\mapsto w_{I,J}(t)$ and $(x,t)\mapsto w_{I,I}(|x|_p,t)$ are uniformly continuous for each $I,J\in V$. Then
there exists a probability transition function $P(t,x;s,\cdot)$, where $(t,x,s)\in [0,T]\times K \times [0,T] $, and $s\leq t$, on the Borel $\sigma$-algebra of $K$, such that the Cauchy problem 
\begin{equation*}
    \begin{cases}
      \frac{\partial u}{\partial t} (x,t)=\mathbf{W}(t)u(x,t)\\
      u(x,s)=g(x)\in C(K_N),
    \end{cases}
\end{equation*}

has a unique solution satisfying: 
\[\mathbb{E}[\varphi(X_t) | X_s\sim u(x)dx]=\int_{K_N \times K_N}\varphi(y)P(t,x;s,dy)u(x)dx=\int_{K_N}\varphi(x)u(x,t)dx\]
In addition, $P(t,x;s,\cdot)$ is the transition function of a strong Markov process.
\end{teo}
\begin{proof}
   Let $t_0>0$ fixed, and let $w_0(x,y)=w(x,y,t_0)$. Consider the integral operator 
\[\textbf{W}^{*}u(x)=\int_{K_N}w_0(x,y)(u(y)-u(x))dy,\]
acting on the domain $C(K_N)$. It is easy to show that $\textbf{W}^{*}$ generates a feller semigroup by means of the Hille-Yosida theorem (see Theorem 17.11, \cite{kallenberg_foundations_2021}), nevertheless for the sake of completeness we include the proof. Since the operator is defined in all the space $C(K_N)$ we only have to show $a)$ the range of $\lambda_0-\textbf{W}^{*}$ is dense in $C(K_N)$ for some $\lambda_0>0$ and $b)$ If for $x_0\in K_N$, $f(x_0)=\sup_{x\in K_n} f(x)>0$, then $\textbf{W}^{*}f(x_0)<0$. For asertion $a)$, the result follow from the boundedness of the operator, since $||\textbf{W}^{*}||\leq \lambda_0$, for some $\lambda_0>0$ implies the existence and boundedness of the operator $(1-\frac{1}{\lambda}\textbf{W}^{*})^{-1}$, which implies $rank \ (1-\frac{1}{\lambda}\textbf{W}^{*})= C(K_N)$. By Proposition \ref{contractionSG}, the time dependent operator \[\textbf{W}^{*}(t)u(x)=\int_{K_N}w(x,y,t)(u(y)-u(x))dy,\]
generates a Feller evolution. That is, following \cite[Thm.\ 2.9]{Schnaubelt2006} there exists a strong Markov process with probability transition $P(t,x;s,dy)$, such that
\[P_{\textbf{W}^{*}}(t,s)u(x)=\int_{K_N}u(y)P(t,x;s,dy),\]
where \(P_{\textbf{W}^{*}}\) is the evolution family attached to $(\textbf{W}^{*}(t),C(K_N))$. The above equality implies that for a given probability measure $\mu(dx)$ the random variable $X_t$ satisfies:
\[\mathbb{E}[\varphi(X_t) | X_s\sim \mu(dx)]=\int_{K_N \times K_N}\varphi(y)P(t,x;s,dy)\mu(dx).\]
A straightforward computation lead to the following identity: 
\[\int_{K_N}(\textbf{W}^{*}(t)f)(x)g(x)dx=\int_{K_N}f(x)(\textbf{W}(t)g)(x)dx.\]
for every $f,g\in C(K_N)$. Moreover, the time-dependent operator $\textbf{W}^{*}(t)$ is the $L^2(K_N)$ adjoint operator of $\textbf{W}(t)$. The above equality implies the dual relation for the evolution families: if $P_{\textbf{W}}(t,s)$ is the evolution family attached to $\textbf{W}$, we have that 
\[\int_{K_N}(P_{\textbf{W}^{*}}(t,s)f)(x)g(x)dx=\int_{K_N}f(x)(P_{\textbf{W}}(t,s)g)(x)dx.\]
Since $u(x,t)=(P_{\textbf{W}}(t,s)u)(x)dx$ is the solution of the Cauchy problem, the above identity implies the desired result: 
\[\mathbb{E}[\varphi(X_t) | X_s\sim u(x)dx]=\int_{K_N \times K_N}\varphi(y)P(t,x;s,dy)u(x)dx=\int_{K_N}\varphi(x)u(x,t)dx\]
\end{proof}

\subsection{The effect of ultrametricity on the evolution process.}

So far we have studied the general properties of the stochastic process attached to time-dependent $p$-adic transition functions. We now study the effect of ultrametricity on the evolution process. In particular, we show how ultrametricity implies a simple computational description of the behavior of the dynamics, allowing us to consider a high number of states without compromising the computational complexity. In order to do that, we will express the solution family $P(t,s)$ in terms of the semigroups attached to the operators $\mathbf{W}(t_0)$. This is achieved by the well-known Trotter-Kato Theorem presented below. 

\begin{prop}\label{uniformConvergence}
Let $\mathbf{W}(t)$ be a family of strongly continuous bounded operators in $C(K)$, and let $P(t,s)$ be its respective evolution family. Then 
\begin{equation}\label{trotterkatoformula}
    \lim_{n\rightarrow \infty} \prod_{k=1}^{n}e^{t/n\mathbf{W}(s+kt/n)}f=P(s+t,s)f
\end{equation}
for all $f\in C(K)$ and uniformly for $s$ and $t$ in compact intervals of $\mathbb{R}$ and $\mathbb{R}_{+}$, respectively. 
\end{prop}
\begin{proof}
\cite[Ch.\ III.5.9, Prop.]{EN2006}.
\end{proof}
Denote by $\textbf{W}^{(0)}$ the matrix representation of $\textbf{W}(t_0)$, for a fixed $t_0>0$, acting on $\mathbb{C}^{|V|}\cong span_{\mathbb{C}}\{\Omega(\absolute{x-I}_p)\}$.  
\begin{prop}[Eigenvalue problem,]\label{somethingKnown}
    The elements of the set: 
\[
\Spec(\textbf{W}^{(0)}) \bigsqcup \{-\gamma_{I,r}: I\in V,r\leq 0 \},
\]
where
\[\gamma_{I,r}=\int_{\mathbb{Z}_p\setminus p^{r-1}\mathbb{Z}_p}w_{II}(|x|_p)dx+p^{r-1}w(p^{r})+\sum_{J \in V} w_{IJ}\]
are the eigenvalues of the operator $\textbf{W}(t_0)$, for fixed $t_0>0$. The corresponding eigenfunctions are given by the following infinite set 
\[
\left\{ \frac{\varphi_\mu}{||\varphi_\mu||_2}: \mu\in \Spec(\textbf{W}_F^{(0)})\right\}
\bigsqcup 
\left\{
 \psi_{rjn}^{I}:I\in V, j\in \{1,...,p-1\}, r\in \mathbb{Z}, r\leq 0, n\in \mathbb{Q}_p/\mathbb{Z}_p\right\},
\]
where the functions $\varphi_\mu(x)$ are defined by: 
\begin{equation}
  \varphi_\mu(x)=\sum_{J\in V(\mathcal{G})} \varphi_\mu(J)\Omega(p^{-N}|x-J|_p)
\end{equation}
and the vector $(\varphi_\mu(J))_{J\in V(\mathcal{G})}$ is an eigenvector of the matrix $\textbf{W}^{(0)}$ corresponding to $\mu\in\Spec(\textbf{W}^{(0)})$, and $ \psi_{rjn}^{I}$ are the Kozyrev functions of the form (\ref{KozyrevWavelet}). 
\end{prop}
\begin{proof}
Since the space $L^2(K_N)$ admit the following decomposition $L^2(K_N)=\mathbb{C}^{N}\oplus \mathcal{L}_0(K),$ we have only to show that every wavelet of the form $\psi^{I}_{rjn}$is an eigenfunction of the operator $\textbf{W}(t_0)$. Indeed, notice that 
\[\textbf{W}(t_0)\psi^{I}_{rjn}(x)=1_{I+\mathbb{Z}_p}(x)\int_{K_N}w_{I,I}(|x-y|_p)(\psi^{I}_{rjn}(y)-\psi^{I}_{rjn}(x))dy-\psi^{I}_{rjn}(x)\sum_{J \in V} w_{IJ}.\]
Define $\hat{\gamma}_{I,r}=\int_{\mathbb{Z}_p\setminus p^{r-1}\mathbb{Z}_p}w_{II}(|x|_p)dx+p^{r-1}w(p^{r})$, then it is known (see for example \cite{ZunigaNetworks2}), that 
\[\int_{K_N}w_{I,I}(|x-y|_p)(\psi^{I}_{rjn}(y)-\psi^{I}_{rjn}(x))dy=-\hat{\gamma}_{I,r}\psi_{rjn}^{I}(x).\]
Therefore, $W_F\psi^{I}_{rjn}(x)=-\gamma_{I,r}\psi^{I}_{rjn}(x)$, as wanted. 
\end{proof}

Let $t_0\geq 0$ be a fixed positive real number. Then the operator $\mathbf{W}(t_0)$ act on the space $L^2(K_N)$ as a direct sum. As shown in the section 2.1 we have the decomposition 
\[L^2(K_N)=\mathbb{C}^{|V|}\oplus \mathcal{L}_0(K_N).\]
For any $f\in L^2(K_N)$ denote by $\hat{f}$ the projection of $f$ in the space $\mathbb{C}^{|V|}$. Then, by \ref{expansionbasis} we have

\[f(x)=\hat{f}+\sum_{\supp(\Psi_{j,I,r}) \subset K_N} C_{j,I,r}\psi_{j,I,r}.\]
The operator $\mathbf{W}(t_0)$ act diagonally on the Kozyrev basis $\psi_{j,I,r}$ as shown in Proposition  \ref{somethingKnown}.   Therefore, the evolution family can be expressed as
\[P(t+s,s)f(x)=\hat{P}(t+s,s)\hat{f}(x)+P(t+s,s)\left(\sum_{\supp(\Psi_{j,I,r}) \subset K_N} C_{j,I,r}\psi_{j,I,r}(x)\right). \]
Where $\hat{P}(t,s)$ is the evolution family attached to the time-dependent matrix $\textbf{W}(t)|_{\mathbb{C}^{|V|}}$. For finite systems, the computation of  $\hat{P}(t,s)$ is a well know matter, for example, numerical methods like the Dyson expansion or the Trotter-Kato formula can be used to approximate $\hat{P}(t,s) $ i.e. the solution of 
\[\frac{d}{dt} u(t) = \textbf{W}(t)|_{\mathbb{C}^{|V|}}, \quad u(0) = u_0\in \mathbb{C}^{|V|}.\]
Solving this equation analytically could be very difficult, and the general case of two basins, that is, when the size of $\textbf{W}(t)|_{\mathbb{C}^{|V|}}$ is two, is already not trivial. However, when the transitions between states are ultrametric, the ultrametric part of the solution
\[P(t+s,s)\left(\sum_{\supp(\Psi_{j,I,r}) \subset K_N} C_{j,I,r}\psi_{j,I,r}(x)\right),\]
has an analytic closed form as shown in the next Theorem. \\

\begin{teo}\label{CauchyProblem}
Let $w(x,y,t) $ be a time-time dependent $p$-adic transition function satisfying the hypothesis of Theorem 2. Then the evolution family of the Cauchy problem (\ref{nAsyst}) (with initial condition at $s$) is given by 
\begin{align*}
P(t+s,s)u_0(x)&=\hat{P}(t+s,s)\hat{u}_0(x)
\\
&+\sum_{\supp(\Psi_{j,I,r})\subset K_N} C_{r,I,j}(s) e^{-\int_{s}^{t}\gamma_{I,r}(\tau)d\tau} \psi_{j,I,r}(x),
\end{align*}
where $t\ge s$, and the coefficients $C_{r,I,j}(s)$ are uniquely determined by the initial condition.
\end{teo}
\begin{proof}
    Without loss of generality, we assume $s=0$ and $\hat{u}_0=0$. Let $n\in \mathbb{N}$. Then, since for all $t_1>0$ and $t_2>0$, the operator $e^{t_1\mathbf{W}(t_2)}$ acts diagonally on the Kozyrev basis, we have:  
    \begin{equation*}
    \begin{split}
        e^{t_0/n\mathbf{W
        }(t_0)}&\cdot e^{t_0/n\mathbf{W}(t_0-t_0/n)}\dots e^{t_0/n\mathbf{W}(t_0/n)}u_0(x)\\&=\sum_{\supp(\Psi_{j,I,r})\subset K_N}C_{r,I,j}(t_0/n)e^{-\frac{t_0}{n}\sum_{k=1}^n\gamma_{I,r}(k\frac{t_0}{n})}\psi_{j,I,r}(x),
    \end{split}
\end{equation*}
Notice that, this is a consequence of the fact that each semigroup act diagonally in a same basis over all times. The result follows by taking the limit when $n\rightarrow \infty$ in virtue of equation \ref{trotterkatoformula}.
\end{proof}

\section{A two basin model and characteristic relaxation.}

In this section, we will study a simple model that consists of two meta-basins each of one consints on sub-basins parametrized $p$-adically. This model can be seen as a time-dependent generalization of the minimalist model of Mauro \cite{Mauro2012Minimalist}. First, we develop further consequences on the Trotter-Kato formula, we then compute explicitly the relaxation process  for an initial distribution on a region $B_{r_0}$, inside the first meta-basin, corresponding, in the case where protein folding data is used, to the unfolded meta-basin, and in the case of glass relaxation parameters with the higher energy meta-basin. \newline

Each of the meta-basins contains, as has already been discussed, an infinitesimal infinity of states (which is an approximation of a continuous-time Markov chain. with a large number of discrete states). The transition functions between the states within each basin are governed by two radial functions, called $w_U$ and $w_N$. The basins will be denoted by $B_U$ and $B_F$, which are two $p$-adic balls that we will assume have a radius equal to $1$. The transitions between $B_U$ and $B_F$ will be governed by the time dependent coefficients \( W_{U\rightarrow F}(t)\) and \( W_{U\rightarrow F}(t) \), respectively.  Let $B_{r_0}\subset B_U$ be a small ball inside the unfolded basin. We aim to describe the characteristic relaxation of this particular region. where, a relaxation process will be understood as the evolution of population (or occupation probability) in the domain of the initial distribution; in this case $B_{r_0}$. 

For this, define the initial condition as $u_0(x)=\frac{1}{p^{-r_0}}1_{B_{r_0}}(x)$. Then, if $u(x,t)$ is the solution of the attached master equation, we will compute $S(t)= \langle p^{-r_0}u_0(x),u(x,t)\rangle$. The initial condition $u_0(x)$ has an expansion on the corresponding eigenbasis of the form
\[u_0(x)=\hat{u_1}(x)+\sum_{Supp \ \psi_{r,j,n} \not  \subset  B_U}C_{r,j,n}\psi_{r,j,n}(x),\]
where $\hat{u}_1(x)=1_{B_{U}}(x)$, and the sum of the right is finite. Note the projection $\hat{u}_1$ follows the master equation attached to the rate matrix $W(t)=[W_{I\rightarrow J}(t)]$.  The solution of the initial value problem \ref{nAsyst} has the form
\[u(x,t)=\hat{u}_1(x,t)+\hat{u}_2(x,t)+\sum_{Supp \ \psi_{r,j,n} \not  \subset  B_U} C_{r,j,n}(s) e^{-\int_{s}^{t}\gamma_{r,U}(\tau)d\tau} \Psi_{r,j,n}(x)\]
where $\hat{u}_1(x,t)=p_1(t)1_{B_{U}}(x)$ and  $\hat{u}_2(x,t)=p_2(t)1_{B_{F}}(x)$. Since $\int_{B_U \cup B_F} u(x,t)=1$ for all times, we have that 
\begin{equation}\label{conservation}
    p_1+p_2=1.
\end{equation}
The evolution of $\hat{u}_1$ follows the master equation attached to $W(t)$:

\[
\frac{dp_1(t)}{dt} = W_{U\rightarrow F}(t)\, p_2(t) - W_{F\rightarrow U}(t)\, p_1(t)
\]
By the conservation condition \ref{conservation}, we derive the equation 

\[
\frac{dp_1(t)}{dt} + \lambda(t)\, p_1(t) = W_{U\rightarrow F}(t)
\]
where $\lambda(t) = W_{U\rightarrow F}(t) + W_{F\rightarrow U}(t) $. This equation can be solved in terms of its integrating factor \(
\mu(t) = \exp\left( \int_{0}^{t} \lambda(\tau)\, d\tau \right)
\), and its given by 
\begin{equation}\label{clasicprob}
    p_1(t) = \frac{1}{\mu(t)} \left( 1 + \int_0^t \mu(s)\,W_{U\rightarrow F}(s)\, ds \right).
\end{equation}
Its worth to mention this gives the solution for the time-dependent version of the classical two-state reaction model. The time-independent case is well known and highly used in transition models and can be solved directly; see for example,\cite{Nolting2005Protein}. On the other hand, the eigenvalues attached to the wavelets $\psi_{r,j,n}(x)$ supported in the basin $B_U$ are given by  
\[-\gamma_{r,U}(t)=-(1-p^{-1})\sum_{j=0}^{-r}p^{-j}w_U(p^{-j},t)-p^{r-1}w_U(p^{r},t)-W_{U\rightarrow F}(t),\]
The solution is 
\[\begin{split}
    u(x,t)= p_1(t)1_{B_{U}}(x)+p_2(t)1_{B_{F}}(x)+\sum_{Supp \ \psi_{r,j,n} \not  \subset  B_U}C_{r,j,n}e^{-\int_{0}^{t}\gamma_{r,U}(\tau)d\tau}\psi_{r,j,n}(x)
\end{split} . \]
Therefore, the characteristic relaxation is given by
\begin{equation}
    \label{relaxtime}
    \begin{split}
   S(t)=  p^{-r_0}p_1(t)+p^{-r_0}\sum_{Supp \ \psi_{r,j,n} \not  \subset  B_U}|C_{r,j,n}|^2e^{-\int_{0}^{t}\gamma_{r,U}(\tau)d\tau}
\end{split} 
\end{equation}
Now, the objective is to analyze the long-term behavior of our system in two possible scenarios. The first involves having constant energy barriers with variable temperature. The second involves having an approximately constant temperature with time-dependent energy barriers.

For the sake of definiteness, we will use the standard Arrhenius relation, so the radial functions are defined by the relation:

\[
w_I(|x|_p, t) = W_I \exp\left( -\frac{U(|x|_p, t)}{k_B T(t)} \right),
\]

where \( I \in \{U, F\} \), \( U(|x|_p) \) is the height of the activation barrier for the transition from state \( y \) to state \( x \), \( k_B \) is the Boltzmann constant, and \( T \) is the temperature. Similarly, the transitions between the basins \( B_U \) and \( B_F \) are determined by the relation

\[
W_{I \to J}(t) = W_{I,J} \exp\left( -\frac{\Delta_{I \to J}U(t)}{k_B T(t)} \right).
\]

The energy barriers inside each basin are parameterized by the $p$-adic radial function.
The heights of the energy barriers may depend on time, as, for example, inside a cell during its life cycle due to interactions with its environment \cite{WPG2013}, where the barriers between the unfolded state and the folded state increase in time during the transition of the interpahse to mitosis.

\subsection{Time-dependent energy barriers and temperature}
Our next goal is to describe the behavior of the probability \ref{clasicprob} and \ref{relaxtime} in two scenarios. First we assume a constant energy landscape while the temperature decreases simulating a cooling using the parameters used in a  glass relaxation model proposed in \cite{Mauro2012Minimalist}. We propose this example since it is well known that the master equation attached to glass relaxation usually depends on a temperature path $T(t)$. On the other hand, motivated by the studies on protein folding and its dependence on temperature, we use the temperature dependent transitions given in \cite{Guo2012Temperature}, to model a protein folding dynamic assuming a linear increase in temperature. 
\subsection{Anomalous relaxation caused by fast cooling}
Assume the energy barriers between the two basins $B_U$ and $B_F$ are time independent:  $\Delta_{I\rightarrow J}U(t)=\Delta_{I\rightarrow J}U$. We now make some general observations of \ref{clasicprob} based on the theory developed in this manuscript. We can rewrite $p_1(t)$ as 
\[p_1(t)=\frac{1}{\mu(t)}+\int_{0}^{t}exp\left(-\int_{s}^{t}\lambda(\tau)d\tau\right)W_{U\rightarrow F}(s)ds,\]
by performing the change of variables $u=t-s$ we obtain
\[p_1(t)=\frac{1}{\mu(t)}+\int_{0}^{t}exp\left(-\int_{t-u}^{t}\lambda(\tau)d\tau\right)W_{U\rightarrow F}(t-u)du. \]
For $u$ sufficiently small (and therefore small $t$), we can take the first order Taylor expansion approximation of $-\int_{t-u}^{t}\lambda(\tau)d\tau$ and $W_{U\rightarrow F}(t-u)$ as functions of $u$. We obtain the following expression 
\[\begin{split}p_1(t)=\frac{1}{\mu(t)}+
    &\frac{W_{U\rightarrow F}(t)}{W_{U\rightarrow F}(t) + W_{F\rightarrow U}(t)}(1-e^{-\lambda(t)t})\\&-\frac{W_{U\rightarrow F}'(t)}{(W_{U\rightarrow F}(t)+ W_{U\rightarrow F}(t))^2}(1-e^{-\lambda(t)t})(1+\lambda(t)t))\end{split}.\]
\begin{figure}[H]
\label{temperaturepath}
    \centering
    \includegraphics[width=\textwidth]{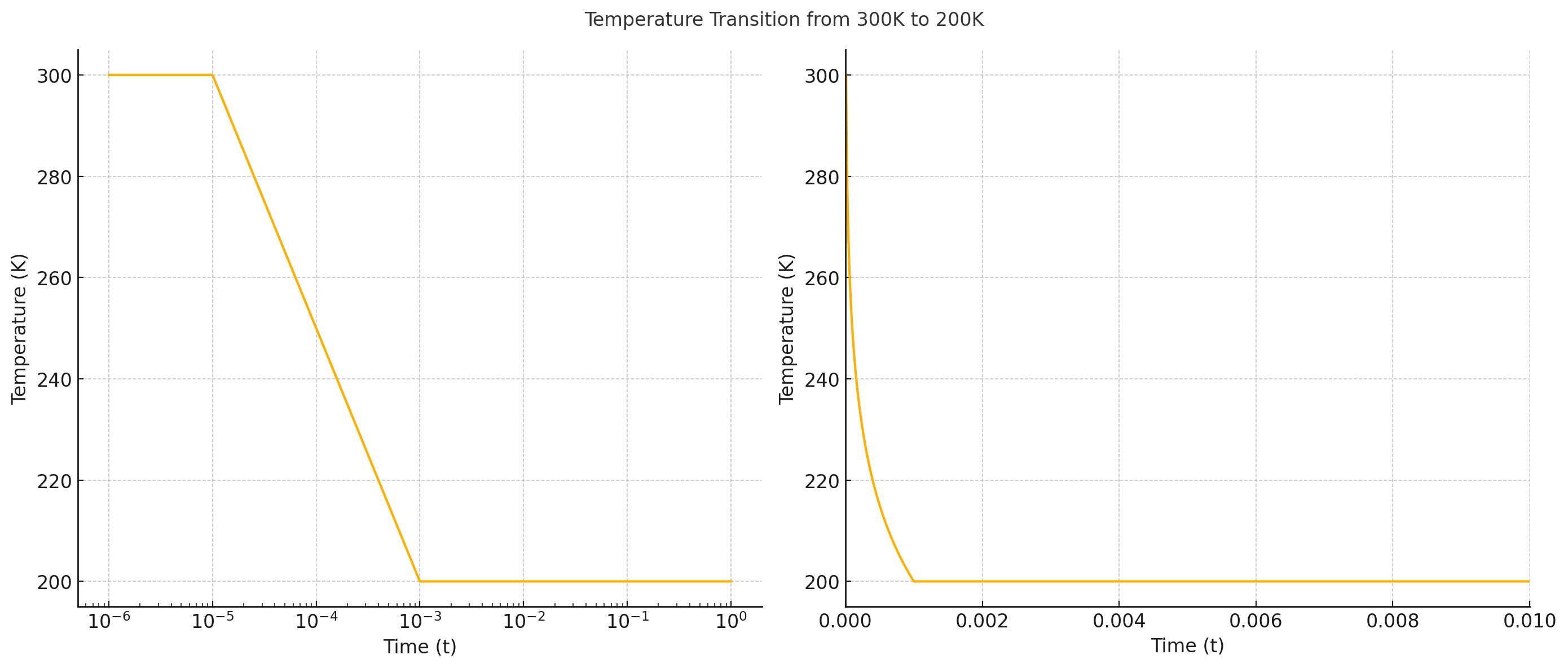}  \caption{Temperature path for temperature drop from $300K$ to $200K$}
    \label{fig:p1_temperature_changes}
\end{figure}

In particular we can make the assumption $W'_{U\rightarrow F}(t)\simeq 0$, since in this interval the function behaves almost constant, that is $W_{u\rightarrow F}(t)\simeq W_{U\rightarrow F}(t_0)$.  We now use Trotter-Katto Theorem to give a further approximation to $p_1(t)$. Define a time interval of analysis $[0,t_0]$. And define the partition of this interval upto a resolution, that is chose $0=a_0<a_1<...<a_n=t_0$. Next evolve the system from $t=0 $ to $t=a_1$ acording to 
\[p_1(t)=\frac{1}{\mu(t_0)}+
    \frac{W_{U\rightarrow F}(t_0)}{W_{U\rightarrow F}(t_0) + W_{F\rightarrow U}(t_0)}(1-e^{-\lambda(t_0)t})\]
which is a good approximation for a small intervals $t_0\in [a_0,a_1]$. Now using the property 2 of Definition \ref{FellerEvolution}, and the Trotter-Kato formula \ref{trotterkatoformula}, we can evolve the system until $t=a_1$, and then change the initial distribution of the system to be equal to $p_1(a_1)$, and evolve the system on $t\in [a_1,a_2]$ according to 
\[p_1(t)=\frac{p_1(a_1)}{\mu(t)}+
    \frac{W_{U\rightarrow F}(t_1)}{W_{U\rightarrow F}(t_1) + W_{F\rightarrow U}(t_1)}(1-e^{-\lambda(t_1)t})\]
for $t_1\in[a_1,a_2]$. We  then evolve the system until $t=a_2$, and repeat. This lead to a recursive expression which is good enough when $W'_{U\rightarrow F}(t)\simeq 0$ in sufficiently small intervals $[a_i,a_{i+1}]$. For the next figure, we set $W_{U\rightarrow J}=W_{J\rightarrow U}=10^{12}Hz$, $\Delta U_{U\rightarrow F} =0.5\  eV$, $\Delta U_{F\rightarrow U} =0.8 \ eV$. $T(t)$ change in different ranges, from an initial $T(t)=300K$ to different sub-temperatures, ranging from $290K$ to $200K$ in a short period of $10^{-5}$ to $10^{-3}$ seconds as shown in Figure 5.

Trotter-Kato approximation give us another important observation. If $\lim T(t)=T_0$, then is clear $p_1$ achieves a stationary distribution, this is a consequence of  \ref{trotterkatoformula}, since for sufficiently large $t$, the semigroup operator $\textbf{W}(t)$ behaves approximately constant, making the factors of the form $e^{t/n\textbf{W}(kt/n)}$ accumulate for $k>k_0$. And therefore, the system behaves as  $e^{t(n-k_0)/n \textbf{W}(k_0t/n)}$. This can be seen in the behavior of $p_1(t)$ in Figure 6. Where the relaxation is "delayed" in the same time scales as the temperature change. We see how a more drastic cooling lead to a more extended delay, but as mentioned before, the system achieve a stationary distribution since the temperature after a $t=10^{-3}$ remains constant. 

\begin{figure}[H]
    \centering
    \includegraphics[width=\textwidth]{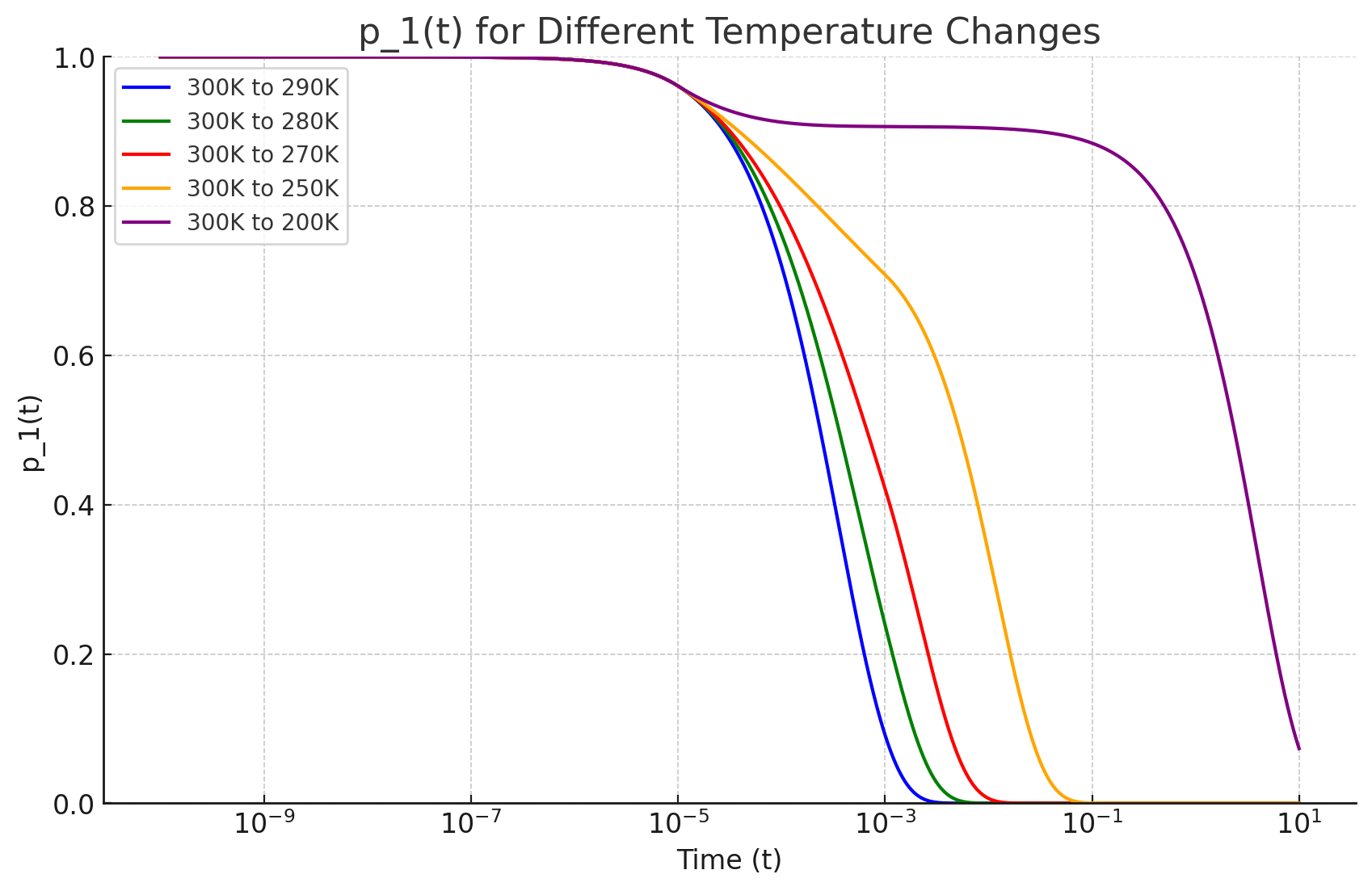}     \caption{Anomalous relaxation caused by super fast cooling .}
    \label{fig:p1_temperature_changes}
\end{figure}

In order to analyze the behavior of $S(t)$ (the characteristic relaxation of the intra-basin $B_{r_0}$, or the survival probability) we now take $W_I=10^{12}Hz$, and we take $r_0=-1$, $p=3$, $U(1)=0.4eV$ and $U(1/3)=0.38eV$. Note that this information is sufficient to calculate \( S(t) \) since only the wavelets with parameter \( r_0 =- 1 \) survive in the expansion of the function \( 1_{B_0} \) with respect to the eigenbasis. Therefore, it is only necessary to explicitly express the radial function \( w_U(|\cdot|_p,t) \) up to the corresponding levels \( |\cdot|_p = 1, 1/p \), i.e, only specify the energy barriers up to this levels.

\begin{figure}[h]
    \centering
    \includegraphics[width=\textwidth]{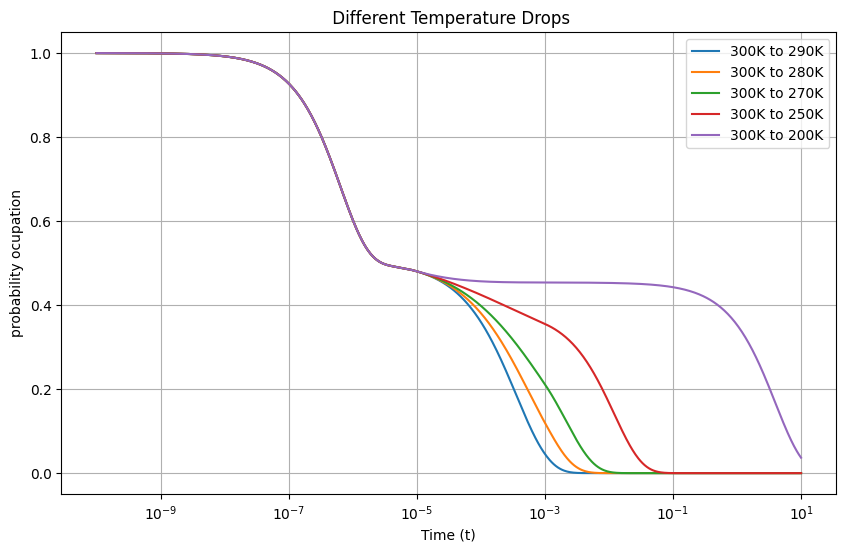} 
    \caption{Occupation probability of intra-basin $B_{r_0}$ by super fast cooling .}
    \label{fig:p1_temperature_changes}
\end{figure}
By using the $p$-adic parametrization we are able to describe the behavior of $S(t)$ intra-basin $B_{r_0}$ without the need to compute an extra approximation, since as shown in Theorem 3, the solution for the non-autonomous problem can be computed directly by the expansion of the initial condition in the eigenbasis. As shown in Figure 7, once again, we see a slowing effect on the dynamics as the temperature decreases, as expected; the relaxation slows down during the cooling period, creating tails of different lengths depending on how drastic the cooling was.  
\subsection{A protein protein folding example with dynamic transition rates.}
In order to implement our model to a protein folding scenario, we use the following expressions for temperature dependence for folding and unfolded transition rates  implemented in\cite{Guo2012Temperature} in order to study the temperature dependence on protein folding in living cells. 

\[
\ln\left(\frac{k_0}{k_f(T)}\right) = \frac{1}{R T} \left( \Delta H_f - T \Delta S_f + \Delta C_p^f \left( T - T_m + T \ln\left( \frac{T_m}{T} \right) \right) \right)
\]

\[
\ln\left(\frac{k_0}{k_u(T)}\right) = \frac{1}{R T} \left( \Delta H_u - T \Delta S_u + \Delta C_p^u \left( T - T_m + T \ln\left( \frac{T_m}{T} \right) \right) \right)
\]

Where the rate prefactor \( k_0 \) is a function of the solvent viscosity and its given by:

\[
k_0 = (10 \, \mu s)^{-1} \left( \frac{\eta(22^\circ C)}{\eta(T)} \right)
\]

Where

\[
\eta(T) = 0.226 + 1.0723 e^{\frac{-(T - 10^\circ C)}{33}}.
\]
The data used in the our model is given also in \cite{Guo2012Temperature} for invitro folding and its given in the following table.

\begin{table}[h]
\centering
\begin{tabular}{llcl}
\hline
\textbf{Parameter} & \textbf{Symbol} & \textbf{Value} & \textbf{Units} \\
\hline
Gas constant & $R$ & $8.314$ & J\,mol$^{-1}$\,K$^{-1}$ \\
Melting temperature & $T_m$ & $312.9$ & K \\
\hline
\multicolumn{4}{l}{\textbf{Folding:}} \\
Enthalpy change & $\Delta H_f$ & $-333 \times 10^3$ & J\,mol$^{-1}$ \\
Entropy change & $\Delta S_f$ & $-1.18 \times 10^3$ & J\,mol$^{-1}$\,K$^{-1}$ \\
Heat capacity change & $\Delta C_p^f$ & $-48\times 10^3$& J\,mol$^{-1}$\,K$^{-1}$ \\
\hline
\multicolumn{4}{l}{\textbf{Unfolding:}} \\
Enthalpy change & $\Delta H_u$ & $337 \times 10^3$ & J\,mol$^{-1}$ \\
Entropy change & $\Delta S_u$ & $0.96 \times 10^3$ & J\,mol$^{-1}$\,K$^{-1}$ \\
Heat capacity change & $\Delta C_p^u$ & $-38\times 10^3$& J\,mol$^{-1}$\,K$^{-1}$ \\
\hline
\end{tabular}
\caption{Thermodynamic parameters used in the study}
\label{tab:thermo_parameters}
\end{table}

\begin{figure}[H]
    \centering
    \includegraphics[width=\textwidth]{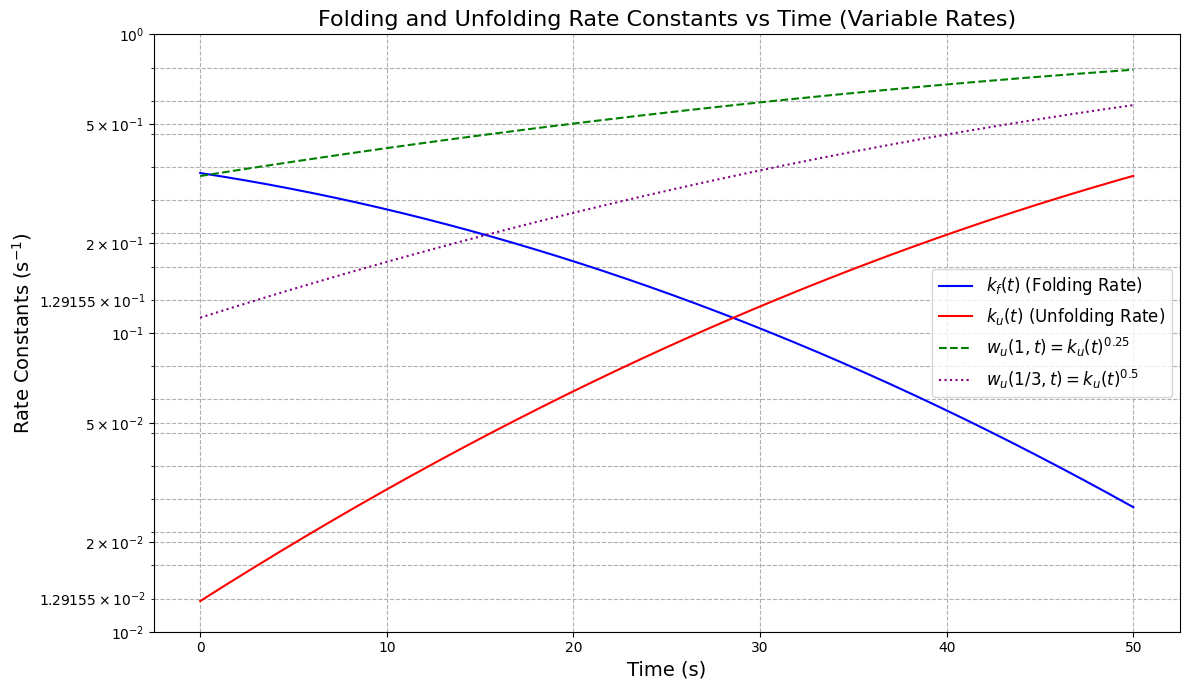}\caption{Rate constants $k_f$ and $k_u$ as a function of time. Additionally the first two values of the radial function $w_u(|\cdot|_p,t)$ are presented .}
    \label{fig:p1_temperature_changes}
\end{figure}

For the interbasin $B_{r_0}$ we use the parameters $r_0=-1$, $p=3$, $w_u(1,t)=k_u^{0.25}$ and $w_u(1/3,t)=k_u^{0.5}$, and we let the temperature increase linearly from 35.85 ° C to 43 ° C in 50 seconds. The selection of these parameters for the radial function is based on the following "toy-model" assumption: throughout the entire process, the first two energy barrier heights within $B_U$ are exactly one-half and one-quarter, respectively, of the energy barrier associated with $k_u$.

The corresponding functions are shown in Figure 8. We see an intersection point at the time where the temperature reaches the melting temperature. On the other hand, the values of the radial function (corresponding to the transition rates within the basin $U$) are most of times greater than $k_f$ and $k_u$, indicating that the unfolded substates interconvert rapidly compared to the folding and unfolding rates. \newline

To have a point of comparison, we show in the following figure the relaxations corresponding to the case when the transition rates are constant; the values of these transition rates are given by the values \( k_f(t) \) at different times \( t \). 

\begin{figure}[H]
    \centering
    \includegraphics[width=\textwidth]{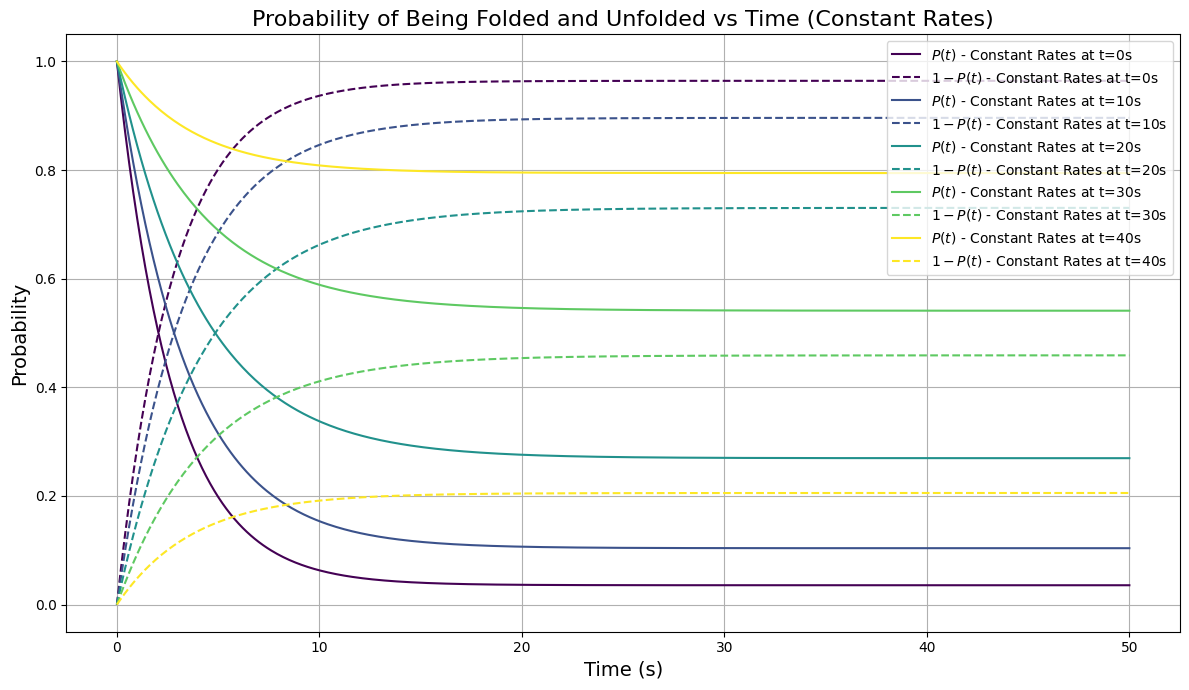}    \caption{Probabilities in the time-independent case.}
    \label{fig:p1_temperature_changes}
\end{figure}
Its clear that in this case, the probabilities follow the usual behavior of the classical two-state transition model (see \cite{Nolting2005Protein}). The time-dependent transition rate case is displayed below. 

\begin{figure}[H]
    \centering    \includegraphics[width=\textwidth]{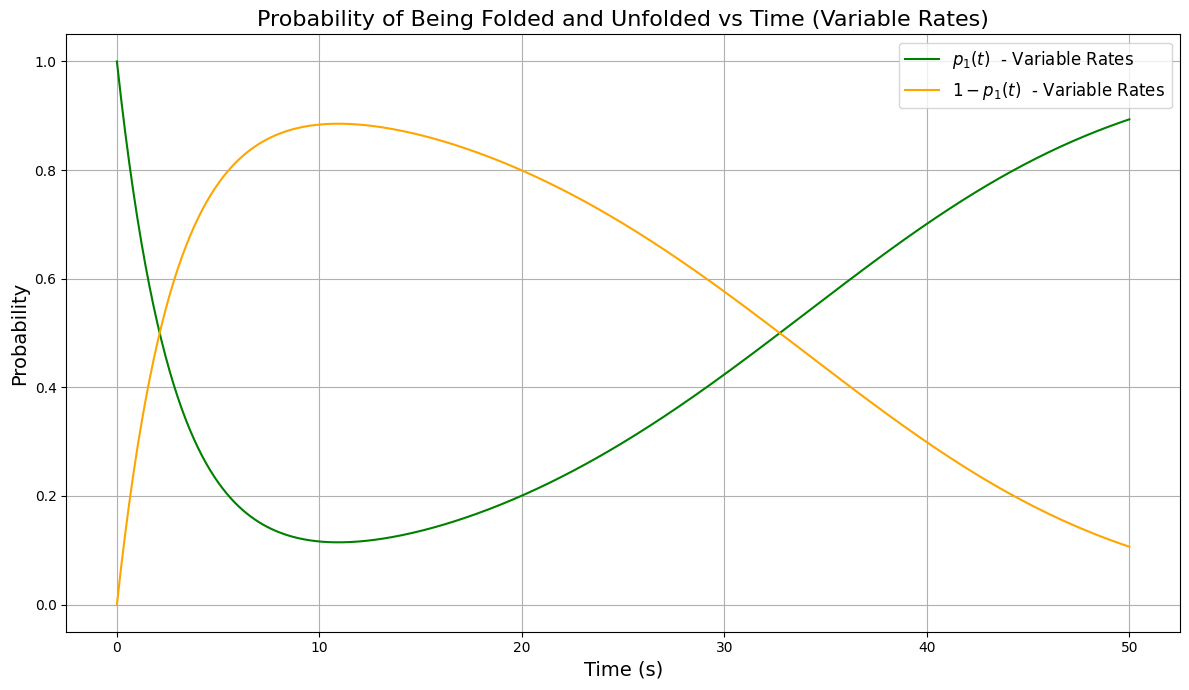}\caption{As before, $p_1$ represent the probability of being in the unfolded state, while $1-p_1(t)$ is the probability of being in the folded state.}
    \label{fig:p1_temperature_changes}
\end{figure}

\begin{figure}[H]
    \centering
    \includegraphics[width=\textwidth]{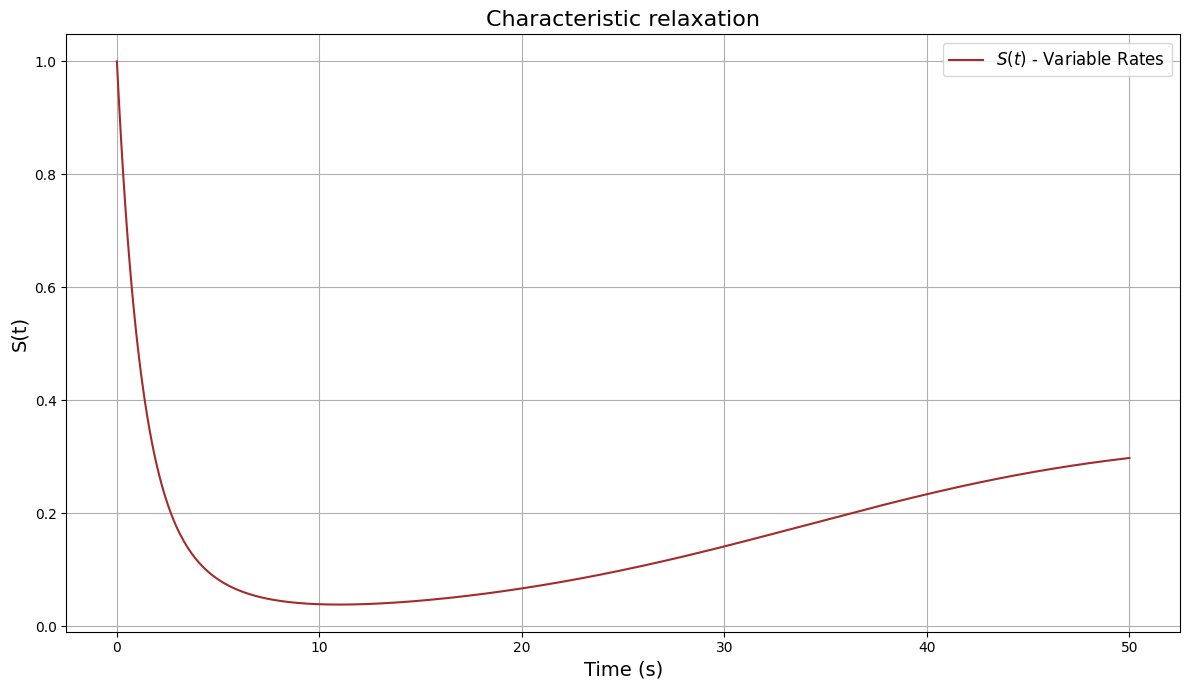}   \caption{Characteristic relaxation of the intrabasin $B_{r_0}$}
    \label{fig:p1_temperature_changes}
\end{figure}

We now explain the results shown in Figure 10 and Figure 11. Initially, the system behaves very similarly to the constant case shown in Figure~9. However, as the temperature increases and approaches the melting temperature, the dynamics change dramatically. In particular, it is interesting to note that seconds before reaching this temperature, the system begins a whiplash effect, where the tendency toward folding reverses, reaching an intersection, as expected, very close to the melting temperature. On the other hand, the unfolded state we analyzed corresponding to the intrabasin \( B_{r_0} \) follows this trend; however, the effect is not completely reversible. As can be observed, the probability \( S(t) \) approaches an equilibrium state because the system can occupy any other unfolded state within basin \( U \). This can be appreciated by noting that at \( t = 50 \), \( S(t) \) is approximately \( \frac{1}{3} \). It is worthwhile to mention that the $p$-adic setting developed in this work allows us to describe an infinite number of relaxation corresponding to smaller balls of the form $B_r$. The hierarchical organization allows us to give an analytic solution of $S(t)$ for any state inside the basin $U$ under the hypothesis, the ultrametricity of such basin prevails during the dynamic. 

\section*{Conclusions}
Our approach demonstrates that \( p \)-adic parametrization and ultrametric analysis are powerful tools for modeling complex systems with dynamic transition rates. By studying complex systems with analytical solutions, we can gain deeper insights into the relaxation processes of physical and biological systems without resorting to computationally intensive simulations. The hierarchical structure inherent in ultrametric spaces allows for an efficient representation of the state space, which is particularly beneficial when dealing with systems that exhibit fractal or self-similar properties.

While our model provides new insights in this direction, it relies on certain assumptions, such as the persistence of ultrametricity throughout the dynamics and the applicability of \( p \)-adic analysis to physical systems. Real-world systems might exhibit deviations from these assumptions due to perturbations or interactions not accounted in the model. Therefore, it is essential to consider these limitations when interpreting the results.

The successful application of our model to both glass relaxation and protein folding suggests that this framework can be extended to other complex systems where hierarchical organization and time-dependent dynamics are prominent. For instance, it could be applied to study of epidemiological models such as in \cite{Khrennikov2021Ultrametric} where the author propose to study the spreading of the COVID-19 by assuming a hierarchic social clustering of population. Adding the role of time-variant dependence on the "social barriers" may lead to new insights.

\section*{Acknowledgements}

I'm indebted to Wilson Z\'u\~niga-Galindo for sharing valuable insights into $p$-adic analysis throughout  my formative years. Patrick Bradley, Leon Nietsche, and David Weisbart are warmly thanked for important discussions and advices. This work is supported by the Deutsche Forschungsgemeinschaft
under project number 469999674.

\bibliographystyle{plain}
\bibliography{biblio}

\begin{thebibliography}{10}

\bibitem{Avetisov2002Padic}
V.~A. Avetisov, A.~H. Bikulov, S.~V. Kozyrev, and V.~A. Osipov.
\newblock $p$-adic models of ultrametric diffusion constrained by hierarchical energy landscapes.
\newblock {\em Journal of Physics A: Mathematical and General}, 35(2):177--189, 2002.

\bibitem{ABK1999}
V.A. Avetisov, A.H. Bikulov, and S.V. Kozyrev.
\newblock Application of p-adic analysis to models of spontaneous breaking of replica symmetry.
\newblock {\em J. Phys. A: Math. Gen.}, 32:8785--8791, 1999.

\bibitem{ABKO2002}
V.A. Avetisov, A.H. Bikulov, S.V. Kozyrev, and V.A. Osipov.
\newblock $p$-adic models of ultrametric diffusion constrained by hierarchical energy landscapes.
\newblock {\em J. Phys. A: Math. Gen.}, 35:177--189, 2002.

\bibitem{ABZ2014}
V.A. Avetisov, A.Kh. Bikulov, and A.P. Zubarev.
\newblock Ultrametric random walk and dynamics of protein molecules.
\newblock {\em Proc. Steklov Inst. Math.}, 285:3--25, 2014.

\bibitem{Bowman2014Markov}
Gregory~R. Bowman, Vijay~S. Pande, and Frank Noé, editors.
\newblock {\em An Introduction to Markov State Models and Their Application to Long Timescale Molecular Simulation}, volume 797 of {\em Advances in Experimental Medicine and Biology}.
\newblock Springer New York LLC, 2014.

\bibitem{nonAutonomousDiffusion}
P.~Bradley and \'A.~Mor\'an Ledezma.
\newblock A non-autonomous $p$-adic diffusion equation on time changing graphs.
\newblock To appear in Reports on Mathematical Physics.

\bibitem{Charbonneau2014Fractal}
Patrick Charbonneau, Jorge Kurchan, Giorgio Parisi, Pierfrancesco Urbani, and Francesco Zamponi.
\newblock Fractal free energy landscapes in structural glasses.
\newblock {\em Nature Communications}, 5:3725, 2014.

\bibitem{jammingultra}
R.~C. Dennis and E.~I. Corwin.
\newblock Jamming energy landscape is hierarchical and ultrametric.
\newblock {\em Phys. Rev. Lett.}, 124:078002, Feb 2020.

\bibitem{DragovichBA2010}
B.~Dragovich and A.~Dragovich.
\newblock A $p$-adic model of dna sequence and genetic code.
\newblock {\em Comput. J}, 53:432--442, 2010.

\bibitem{DXKM2021}
B.~Dragovich, A.Yu. Khrennikov, S.V. Kozyrev, and N.Ž. Mišić.
\newblock $p$-adic mathematics and theoretical biology.
\newblock {\em Biosystems}, 199:104288, 2021.

\bibitem{EN2006}
K.-J. Engel and R.~Nagel.
\newblock {\em One-Parameter Semigroups for Linear Evolution Equations}.
\newblock Graduate Textes in Mathematics 194. Springer, New York, 2000.

\bibitem{Fattorini1983}
H.O. Fattorini.
\newblock {\em The Cauchy Problem}.
\newblock Encyclopedia of Mathematics and its Applications 18. Addison-Wesley, 1983.

\bibitem{Frauenfelder2010Physics}
H.~Frauenfelder, S.~S. Chan, and W.~S. Chan, editors.
\newblock {\em The Physics of Proteins}.
\newblock Springer-Verlag, 2010.

\bibitem{Frauenfelder2003Myoglobin}
H.~Frauenfelder, B.~H. McMahon, and P.~W. Fenimore.
\newblock Myoglobin: the hydrogen atom of biology and paradigm of complexity.
\newblock {\em Proceedings of the National Academy of Sciences}, 100(15):8615--8617, 2003.

\bibitem{Frauenfelder1991Energy}
H.~Frauenfelder, S.~G. Sligar, and P.~G. Wolynes.
\newblock The energy landscape and motions of proteins.
\newblock {\em Science}, 254:1598--1603, 1991.

\bibitem{templandscape2}
H.~Gelman, M.~Platkov, and M.~Gruebele.
\newblock Rapid perturbation of free-energy landscapes: From in vitro to in vivo.
\newblock {\em Chemistry - A European Journal}, 18:6420--6427, 2012.

\bibitem{Guo2012Temperature}
Minghao Guo, Yangfan Xu, and Martin Gruebele.
\newblock Temperature dependence of protein folding kinetics in living cells.
\newblock {\em Proceedings of the National Academy of Sciences}, 109(44):17863--17867, 2012.

\bibitem{Husic2018Markov}
Brooke~E. Husic and Vijay~S. Pande.
\newblock Markov state models: From an art to a science.
\newblock {\em Journal of the American Chemical Society}, 140(7):2386--2396, 2018.

\bibitem{kallenberg_foundations_2021}
Olav Kallenberg.
\newblock {\em Foundations of Modern Probability}.
\newblock Probability Theory and Stochastic Modelling. Springer Cham, 3 edition, 2021.

\bibitem{Khrennikov2021Ultrametric}
Andrei Khrennikov.
\newblock Ultrametric diffusion equation on energy landscape to model disease spread in hierarchic socially clustered population.
\newblock {\em Physica A: Statistical Mechanics and its Applications}, 583:126284, 2021.

\bibitem{XK2018b}
A.Yu. Khrennikov and A.N. Kochubei.
\newblock $p$-adic analogue of the porous medium equation.
\newblock {\em J. Fourier Anal. Appl.}, 24:1401--1424, 2018.

\bibitem{Mauro2021Materials}
John~C. Mauro.
\newblock {\em Materials Kinetics}.
\newblock Elsevier, 2021.

\bibitem{Mauro2007Metabasin}
John~C. Mauro, Roger~J. Loucks, and Prabhat~K. Gupta.
\newblock Metabasin approach for computing the master equation dynamics of systems with broken ergodicity.
\newblock {\em The Journal of Physical Chemistry A}, 111(32):7957--7965, 2007.

\bibitem{Mauro2012Minimalist}
John~C. Mauro and Morten~M. Smedskjaer.
\newblock Minimalist landscape model of glass relaxation.
\newblock {\em Physica A: Statistical Mechanics and its Applications}, 391(12):3446--3459, 2012.

\bibitem{Nolting2005Protein}
Bengt Nolting.
\newblock {\em Protein Folding Kinetics: Biophysical Methods}.
\newblock Springer, 2nd edition, 2005.

\bibitem{Peliti2021Stochastic}
Luca Peliti and Simone Pigolotti.
\newblock {\em Stochastic Thermodynamics: An Introduction}.
\newblock Princeton University Press, 2021.

\bibitem{Schnaubelt2006}
R.~Schnaubelt.
\newblock {VI.9} {Semigroups} for nonautonomous {Cauchy} problems.
\newblock In {\em One-Parameter Semigroups for Linear Evolution Equations}, Graduate Textes in Mathematics 194. Springer, New York, 2000.

\bibitem{vanCasteren2011}
J.A. {van Casteren}.
\newblock {\em Markov Processes, {Feller} Semigroups and evolution equations}.
\newblock Series on Concrete and Applicable Mathematics – Vol.12. World Scientific, 2011.

\bibitem{WPG2013}
A.J. Wirth, M.~Platkov, and M.~Gruebele.
\newblock Temporal variation of a protein folding energy landscape in the cell.
\newblock {\em J. Am. Chem. Soc.}, 135:19215--19221, 2013.

\bibitem{ZunigaNetworks}
W.~{Z\'{u}\~{n}iga-Galindo}.
\newblock Reaction-diffusion equations on complex networks and {Turing} patterns, via $p$-adic analysis.
\newblock {\em Journal of Mathematical Analysis and Applications}, 491(1):124239, 2020.

\bibitem{ZunigaNetworks2}
W.A. Zúñiga-Galindo.
\newblock Ultrametric diffusion, rugged energy landscapes and transition networks.
\newblock {\em Physica A}, 597:127221, 2022.

\end{thebibliography}

\end{document}